\documentclass[ijoo,sglanonrev]{informs4}

\usepackage{eqndefns-left} 
\RequirePackage{tgtermes}
\RequirePackage{newtxtext}
\RequirePackage{newtxmath}
\RequirePackage{bm}
\RequirePackage{endnotes}

\OneAndAHalfSpacedXII 

\usepackage{algorithm}
\usepackage{algpseudocode}
\usepackage{tikz}
\usepackage{bbm}
\usepackage{placeins}

\usepackage{natbib}
 \bibpunct[, ]{(}{)}{,}{a}{}{,}%
 \def\bibfont{\small}%
\EquationsNumberedThrough    

\TheoremsNumberedThrough     
\ECRepeatTheorems  %

\MANUSCRIPTNO{IJOO-0001-2026.00}

\begin{document}



\RUNAUTHOR{Mantegna et al.}

\RUNTITLE{Uncertainty-Aware Grid Planning in the Real World}

\TITLE{Uncertainty-Aware Grid Planning in the Real World: A Method Enabling Large-Scale, Two-Stage Adaptive Robust Optimization for Capacity Expansion Planning}


\ARTICLEAUTHORS{%
\AUTHOR{Gabriel Mantegna}
\AFF{Department of Mechanical and Aerospace Engineering and the Andlinger Center for Energy and the Environment,
Princeton University, \EMAIL{gabe.mantegna@princeton.edu}}

\AUTHOR{Emil Dimanchev}
\AFF{Andlinger Center for Energy and the Environment,
Princeton University}

\AUTHOR{Filippo Pecci}
\AFF{RFF-CMCC European Institute on Economics and the Environment}

\AUTHOR{Neha Patankar}
\AFF{School of Systems Science and Industrial Engineering, Binghamton University}

\AUTHOR{Jesse Jenkins}
\AFF{Department of Mechanical and Aerospace Engineering and the Andlinger Center for Energy and the Environment,
Princeton University}

} 

\ABSTRACT{%
Capacity expansion models are frequently used to inform multi-billion dollar grid infrastructure decisions, a context in which there is significant uncertainty surrounding the future need for and performance of such infrastructure. However, despite much academic literature on the topic, virtually no grid planning processes use capacity expansion models that endogenously consider uncertainty, an oversight which frequently leads to short-sighted infrastructure decisions. This is partially due to a technology transfer gap, but it is also due to a lack of methods that work at large scale. In this paper we introduce a method for endogenizing uncertainty into capacity expansion models, a variant of adaptive robust optimization, that addresses this gap. We apply the method to a real-world capacity expansion planning problem, that of the State of California, and compare its performance to that of traditional adaptive robust optimization. We find that both the traditional method and our method identify increased transmission investment as a key lever for increasing robustness and adaptability, while helping to avoid downside risks that current deterministic planning processes may be exposing ratepayers to. Our method performs similarly to the traditional method in terms of outcomes, while significantly reducing computational complexity, making it scalable to real-world planning problems.
}%

\FUNDING{This research was supported by Sonoma Clean Power, Peninsula Clean Energy, and the Maeder Graduate Fellowship in Energy and the Environment at Princeton University.}




\maketitle


\section{Introduction}\label{sec:Intro}

\subsection{Motivation}\label{subsec:Motivation}

Capacity expansion models are frequently used for planning multi-billion dollar electric grid infrastructure investments around the world, including both power generation and transmission assets. These models are used in decision environments where there is abundant uncertainty about how the future might play out, including future changes in demand, resource availability and costs, policy, and generator retirements. The growth of data centers, electrification of transportation, heat and other end-uses, and the commercialization of nascent technologies such as enhanced geothermal, advanced nuclear, and long-duration storage add to the uncertainty facing planners today.

Despite the abundance of uncertainty that is certain to impact the future performance of these investments, electricity system capacity expansion models that endogenously consider uncertainty are rarely used in practice. This is in spite of an extensive academic literature that presents methods to accomplish this goal (surveyed below) and evidence that endogenously considering uncertainty is crucial to developing investment plans that are flexible and adaptable \citep{munoz_engineering-economic_2014,van_der_weijde_economics_2012}. Models that endogenously consider uncertainty (e.g., via stochastic or robust optimization)  vary in their outcomes relative to deterministic models that assume the future is static. Sometimes, uncertainty-aware models recommend a ``wait-and-see'' approach when the future is highly uncertain and decisions are best made once key uncertainties are resolved, and other times, these models recommend an ``insurance premium'' approach in which near-term investments are made in flexible infrastructure that is likely to be needed in many possible realizations of the future.

Unfortunately, neither strategy will be identified by deterministic modeling with scenarios, a common approach to address uncertainty employed by electricity infrastructure planners today. In this approach, deterministic models are run with multiple scenarios to represent future uncertainty. Such scenario-assisted processes can help provide decision-makers with information about how near-term decisions would change under different scenarios for how the future could play out, but crucially, they are of limited usefulness in helping to identify investments that do well across a range of possible futures. There are some simple process improvements available that are sometimes used: \textit{stress-testing} investment decisions across multiple possible futures can help to inform decision-makers about the performance of investments across multiple scenarios, and \textit{heuristics} that combine information from scenarios can help to identify ``least-regrets'' investment decisions that are likely to be needed no matter how the future pans out \citep{mantegna_integrated_2026}. These methods are used by some transmission planning entities such as MISO \citep{midcontinent_independent_system_operator_miso_2024} and ENTSO-E \citep{sanchis_europes_2015}. However, the fundamental limitation of scenario-based modeling is that each model run assumes that decisions are made with certainty as to how the specific scenario will play out, an assumption that is clearly false. Deterministic scenario runs are thus likely to miss investment options that are not optimal in any one scenario, but do well across many scenarios \citep{munoz_engineering-economic_2014}. Given this limitation, it is likely that many current electricity generation and transmission planning processes are missing out on significant welfare-increasing investment options that have the potential of lowering risks for electricity consumers (ratepayers) in the long run.

Part of the reason for this gap between industry practice and academic methods is a technology transfer lag. Most commercial software tools used for capacity expansion planning do not have uncertainty-aware algorithms built in, except for with regards to weather uncertainty, which is commonly endogenized due to the availability of empirical distributions. In hydropower-dominant systems, planners frequently employ uncertainty-aware algorithms such as Stochastic Dual Dynamic Programming to help with hydro scheduling \citep{maceiral_twenty_2018}, but these methods are similarly focused on weather-related uncertainty. However, there is also a practical computational challenge: most real-world capacity expansion problems are so-called ``massive'' LPs that take hours if not days to run with state-of-the-art commercial solvers, meaning there is limited appetite for adopting uncertainty-aware methods that further expand model complexity and could be orders of magnitude slower to solve. This challenge is exacerbated by the fact that some uncertainty-aware methods, such as stochastic optimization, require the assignment of probabilities to scenarios in a context where defensible probabilities are difficult to determine, which could lead planning processes to become even more contested. To our knowledge, only a handful of real-world planning processes (including Portland General Electric's Integrated Resource Planning process), have used a capacity expansion model that endogenizes non-weather uncertainty sources \citep{portland_general_electric_integrated_2019}. (The PGE IRP process is also a good example of how stochastic methods can be used alongside deterministic methods within a strong scenario analysis framework, as discussed in \cite{mantegna_integrated_2026}.) 

In this paper, we present a novel algorithm for uncertainty-aware capacity expansion planning that is explicitly designed to be computationally tractable  enough to use for real-world, large-scale problems. Our method is an adaptation of two-stage adaptive robust optimization that leverages problem structure to enable a much more scalable formulation than traditional methods. It also has the advantage (like any robust optimization-based approach) of not requiring assignment of probabilities in contexts where defensible probabilities are difficult to determine; rather, robust optimization relies on plausible ranges for uncertain parameters, which may be easier to determine in real-world planning contexts. After a brief literature review, we present our algorithm, as well as a case study where our method is applied to a full-scale capacity expansion problem that replicates the real-world capacity expansion planning model currently used by the state of California \citep{california_public_utilities_commission_inputs_2023}. California is chosen for the case study given the state's current planning context, in which a deterministic modeling framework is used to make generation and transmission planning decisions in an environment of significant uncertainty.

\subsection{Literature Review}\label{subsec:LitReview}

The literature on uncertainty-aware capacity expansion planning can in general be segregated into two buckets: stochastic methods, which require assigning probabilities to scenarios and seek to minimize expected costs, and robust methods, which represent uncertainty via plausible parameter ranges rather than probabilities and seek to minimize costs under worst-case outcomes.

\subsubsection{Stochastic methods}

Several foundational papers published over a decade ago were some of the first to apply stochastic programming -- long used in other domains such as finance -- to electricity system capacity expansion planning. \cite{munoz_engineering-economic_2014} and \cite{van_der_weijde_economics_2012} showed that stochastic capacity expansion modeling, specifically in the context of transmission expansion decisions, can result in lower-cost portfolios in the long term, with cost savings on the same order of magnitude as the near-term transmission investments themselves. Similarly, \cite{jin_modeling_2011} was one of the first to apply stochastic methods to capacity expansion planning, although it was focused on demand and fuel price uncertainty, as opposed to the prior two papers which also incorporated other macro-scale uncertainty sources such as future policy changes. In recent years, these methods have been improved on and extended through the development of methods for solving large-scale stochastic problem instances, such as Benders decomposition \citep{zhang_integrated_2025, pecci_regularized_2025}, progressive hedging \citep{watson_progressive_2011}, and projective hedging \citep{eckstein_projective_2025}. These algorithms allow for the solution of somewhat large-scale stochastic problems, though they are still significantly slower than their deterministic counterparts.

An important note on stochastic methods applied to capacity expansion planning is that often, only weather-driven sources of uncertainty are included (resource availability, and weather-driven load variations). This is the case with some high-profile stochastic models used in practice as described in \cite{zampara_capacity_2025}. These sources of uncertainty are important to include, since they impact the performance of investment decisions. However, a commonly overlooked fact is that these models are not fundamentally different from capacity expansion models that endogenously consider a wide range of operating conditions using representative periods \citep{merrick_representation_2016, poncelet_selecting_2017} and that parametrize probabilistic reliability constraints through reserve margins and capacity accreditation \citep{mantegna_electric_2025}. Thus, in practice it is common for practitioners to use ``stochastic'' models that do not represent the most important forms of long-term uncertainty that could affect the performance of investment decisions.

\subsubsection{Robust methods}

Robust optimization is another class of methods that has been applied to capacity expansion problems in the academic literature. Robust optimization was first developed in the late 90s and early 00s \citep{ben-tal_robust_2000, bertsimas_price_2004} and differs from stochastic optimization in that it uses plausible ranges that uncertain parameters can be expected to take, rather than probabilities, to represent uncertainty. Optimal solutions are then found based on ``worst cases'' within a desired uncertainty set. This method has been applied in the energy systems planning context by papers such as \cite{patankar_using_2022}, although in this case only uncertainties in resource and fuel costs were included. Two-stage adaptive robust optimization, an important extension which allows for the adaptation of later decisions in response to realizations of uncertainty, was first developed and applied by \cite{zeng_solving_2013} and applied to power sector problems by \cite{bertsimas_adaptive_2013} in the context of the unit commitment and economic dispatch problem, which is similar in structure to the capacity expansion planning problem but concerns day-ahead scheduling of thermal power plants in the face of forecast error and generator outages. This method has been applied to the capacity expansion planning problem by papers such as \cite{moreira_reliable_2017}, \cite{moreira_climateaware_2021}, and \cite{bernecker_adaptive_2025}. Crucially, these latter papers only model investment decisions in the first stage, and operational decisions in the second stage, an important simplification (likely made for tractability) that means these algorithms are of limited usefulness for real-world planning problems, where multiple investment stages are nearly always included. This is not just an oversight, but rather a result of the inherent complexity of two-stage adaptive robust optimization: the method involves maximization of a convex function, making the optimization problem NP-hard (in layman's terms, this means that the formulation does not work well for large-scale problems). Given real-world deterministic planning problems with multiple investment stages can already take hours if not days to solve, it is highly unlikely that traditional two-stage adaptive robust optimization could be used to solve problems of this scale.

\subsection{Purpose and Contribution}\label{subsec:Purpose}

The purpose of this paper is to present an optimization formulation for capacity expansion planning that endogenizes important sources of uncertainty \textit{and} is tractable enough that it can be used for real-world, large-scale planning problems. We use a reformulation of two-stage adaptive robust optimization that takes advantage of partial convexities in the planning problem and can thus be solved as a linear program (LP). Following the introduction of our method, we demonstrate the performance of the method in a case study in which we apply the method to a real-world, full-scale planning problem that replicates the capacity expansion planning model used by the state of California known as RESOLVE. We first develop an exact representation of the RESOLVE model using the open-source GenX model \citep{jenkins_enhanced_2017, bonaldo_genx_2024} and demonstrate the benchmarking of our instance, then apply the robust planning method and discuss its performance. The remainder of the paper is structured as follows:

\begin{itemize}
    \item Section \ref{sec:StdFormulation} introduces the existing deterministic and two-stage robust formulations that our formulation is based on.
    \item Section \ref{sec:SplitBudgetFormulation} introduces our enhanced formulation.
    \item Section \ref{sec:CaseStudy} presents the real-world case study and introduces the experimental design (\ref{subsec:ExpDesign}), and then discusses the case study results, split into a presentation of the impacts of endogenizing uncertainty using standard two-stage ARO without our method (\ref{subsec:MainResults}) and a comparison of the standard formulation with our method (\ref{subsec:ComparisonResults}).
    \item Section \ref{sec:Conclusion} provides concluding remarks.
\end{itemize}

\section{Deterministic and two-stage robust capacity expansion}\label{sec:StdFormulation}

\subsection{Deterministic capacity expansion}\label{subsec:Deterministic}

The standard deterministic capacity expansion problem can be written in general form as follows, with $x_1$ and $y_1$ representing first-stage investment and operational decisions respectively, and $x_2$ and $y_2$ representing second-stage (also known as recourse) investment and operational decisions:

\begin{align} \label{cem_det}
    \min_{x_1,y_1,x_2,y_2} & c_1^Tx_1 + b_1^Ty_1 + c_2^Tx_2 + b_2^Ty_2 \\
    \text{subject to:} & \notag \\
    & A_1x_1 + B_1y_1 \leq g_1 \notag \\
    & A_2x_2 + B_2y_2 \leq g_2 + Cx_1 \notag
\end{align}

Investment decisions include generation and transmission investments, and operational decisions represent generator dispatch and transmission flows to meet load balancing and other balancing constraints. This two stage form is used to highlight that some subset of capacity expansion decisions can be considered ``here-and-now'' decisions, which go in the first stage, while the rest of the decisions can be considered ``wait-and-see'' decisions, which go in the second stage. The feasible space of second stage decisions is impacted by the first stage investment decisions as seen by the right-hand-side of the second constraint. In the case study of this paper, the two stages each consist of one planning period, meaning the investment and operational decisions for one target future year, but they could also consist of multiple planning periods.

\subsection{Two-stage adaptive robust capacity expansion}\label{subsec:ARO}

Using this same form, we can write the two-stage adaptive robust version of this problem as:

\begin{align} \label{aro}
    \min_{(x_1,y_1) \in \Psi} \ \big( & c_1^T x_1 + b_1^Ty_1 + \max_{u \in U} \min_{(x_2,y_2) \in \Omega(x_1,u)} c_2(u)^Tx_2 + b_2(u)^Ty_2 \big) \\
    \text{where:} & \notag \\
    \Psi &= \{ x_1 \in \mathbb{R}^m_+, y_1 \in \mathbb{R}^n : A_1x_1+B_1y_1 \leq g_1\} \notag \\
    U &= \{ u \in \mathbb{R}^p : \| u \|_0 \leq \Gamma, \| u \|_{\infty} \leq 1 \} \notag \\
    \Omega(x_1,u) &= \{ x_2 \in  \mathbb{R}^m_+, y_2 \in  \mathbb{R}^n : A_2x_2+B_2y_2 \leq g_2(u) + Cx_1 \} \notag
\end{align}

Where $m$ is the number of investment variables, $n$ the number of operational variables, and $p$ the number of uncertain parameters. In this formulation, the innermost minimization represents optimal recourse decisions under an uncertainty realization $u$, while the outer maximization identifies the worst-case realization $u$ for any set of first-stage decisions. This form utilizes a so-called ``budget'' uncertainty set, originally introduced by \cite{bertsimas_price_2004}, in which $\Gamma$ uncertain parameters are allowed to deviate from their nominal, or deterministically defined, values. Increasing $\Gamma$ has the effect of increasing the ``robustness'' of the solution to uncertainty in the input parameters  \citep{bertsimas_robust_2022}. In the current context of capacity expansion, higher $\Gamma$ corresponds to planning for more simultaneous adverse deviations from the nominal or expected conditions (e.g., higher load, higher costs, or lower resource availability). Note that $c_2(u)$ and the other functions of $u$ are linear mappings from values of u, the uncertainty set expressed in $[-1,1]^p$ space, to values of $c_2$ (and the other uncertain parameters respectively) that vary between the lower bounds $\overline{c}_2-\hat{c}_2$ and upper bounds $\overline{c}_2+\hat{c}_2$ depending on the respective value of $u$. These are linear functions but are written in this form for simplicity, to avoid definition of auxiliary matrices or breaking up $u$ into multiple variables.

Problem (\ref{aro}) is non-convex because of the cardinality constraint $\|u\|_0 \leq \Gamma$. Following \cite{bertsimas_price_2004}, this is often relaxed to an $l_1$-norm constraint to obtain a convex uncertainty set:

\begin{equation}
    U_{cvx}=\{u \in \mathbb{R}^p : \|u\|_1 \leq \Gamma, \|u\|_{\infty} \leq 1 \}
\end{equation}

However, this relaxation is only guaranteed to be tight in single-stage problems, and not in two-stage adaptive robust problems, a technicality that will become relevant later in this section when we explore logical extensions of our novel formulation.

Another special case of (\ref{aro}) that will be relevant later is the case of full recourse, in which the second stage problem is feasible for all first stage feasible points $x_1$ and $y_1$, and all uncertainty realizations $u$:

\begin{equation}
    \Omega(x_1,u) \neq \emptyset \  \forall \ (x_1,y_1) \in \Psi, u \in U
\end{equation}

In the capacity expansion context, this full recourse assumption is easily enforced through the introduction of slack variables in the balancing constraints.

The two-stage adaptive robust formulation has the advantage of ensuring that first stage decisions are ``robust'' to realizations of uncertainty within the uncertainty set, while allowing later decisions to be ``adapted'' after the uncertainty is realized. In this way, important sources of uncertainty can be taken into account when making near-term decisions, while accounting for the fact that corrective actions can be made later, thus preventing solutions from being overly conservative. As previously mentioned, this formulation also has the advantage of not requiring the definition of probabilities, which can be useful in applied decision-making contexts where assigning arbitrary probabilities can be contentious.

Because the uncertainty set is convex and the inner minimization yields a convex value function, the outer maximization becomes a non-convex optimization problem, making the overall problem NP-hard. It is usually approximately solved via a column-and-constraint generation algorithm, as in \cite{zeng_solving_2013}, \cite{bertsimas_adaptive_2013}, and \cite{moreira_reliable_2017}, which is a form of decomposition method in which ``worst case'' scenarios are iteratively added to the master problem. However, this algorithm does not scale to large-scale, real-world capacity expansion instances, given these real-world problems are already so-called ``massive'' LPs that take hours or days to solve. This point is further evidenced by the fact that existing studies applying adaptive robust optimization to capacity expansion planning have implemented the algorithm only on test systems (such as the IEEE 118-bus case as in \cite{moreira_reliable_2017}) rather than real-world, full-scale planning problems.

\section{Split uncertainty budget formulation}\label{sec:SplitBudgetFormulation}

To enable two-stage adaptive robust capacity expansion problems to be solved at large scale, we propose a formulation that uses a split uncertainty budget in order to leverage partial convexities, that can be solved as a linear program (LP). We begin by presenting our formulation, and then proceed to prove that it can be solved as an LP, with the help of two lemmas. Our formulation is as follows:

\begin{align} \label{aro_split}
    \min_{(x_1,y_1) \in \Psi} \ \big( & c_1^T x_1 + b_1^Ty_1 + \max_{u_{RHS} \in U_{RHS}} \max_{u_c \in U_c} \min_{(x_2,y_2) \in \tilde{\Omega}(x_1,u_{RHS})} c_2(u_c)^Tx_2 + b_2(u_c)^Ty_2 \big) \\
    \text{where:} & \notag \\
    \Psi &= \{ x_1 \in \mathbb{R}^m_+, y_1 \in \mathbb{R}^n : A_1x_1+B_1y_1 \leq g_1\} \notag \\
    U_{RHS} &= \{ u_{RHS} \in \mathbb{R}^q : \| u_{RHS} \|_0 \leq \Gamma_{RHS}, \| u_{RHS} \|_{\infty} \leq 1 \} \notag \\
    U_c &= \{ u_c \in \mathbb{R}^r : \| u_c \|_1 \leq \Gamma_c, \| u_c \|_{\infty} \leq 1 \} \notag \\
    \tilde{\Omega}(x_1,u_{RHS}) &= \{ x_2 \in  \mathbb{R}^m_+, y_2 \in  \mathbb{R}^n : A_2x_2+B_2y_2 \leq g_2(u_{RHS}) + Cx_1 \} \notag \\
    \tilde{\Omega}(x_1,u_{RHS}) &\neq \emptyset \  \forall \ (x_1,y_1) \in \Psi, u_{RHS} \in U_{RHS} \notag
\end{align}

In this formulation, the uncertainty set is split into two sets: one for cost coefficients in the objective function ($U_c$), and another for second-stage right-hand-side constraint coefficients, which are represented with scenarios ($U_{RHS}$). This has the impact of modeling a \textit{subset} of the full uncertainty set $U$, since the uncertainty realizations modeled in the split-budget formulation will always have some cost-based uncertainties deviating from their nominal values at optimality, and some right-hand-side uncertainties deviating from their nominal values, whereas the full uncertainty set $U$ may include realizations at optimality where only uncertainties in one or the other category deviate from their nominal values. This has the impact of making the solutions slightly less robust compared to the full uncertainty set formulation, although they are of course guaranteed to be more robust compared to deterministic formulations in which uncertainty is not endogenized. This split-budget formulation is not necessarily an inferior formulation, though, since it might be a reasonable assumption on the part of the modeler that both some cost and non-cost uncertain parameters might be expected to deviate from their nominal values.

Note that, in this formulation, $U_c$ uses the $l_1$-norm relaxation of the budget uncertainty set, which is exact for single-stage problems \citep{bertsimas_price_2004}, but constitutes a relaxation for the present problem as compared to a $l_0$-norm, cardinality-constrained uncertainty set. However, we note that the $l_1$-norm relaxation is frequently used in two-stage ARO for tractability, as for example in \cite{bertsimas_adaptive_2013}. We keep the $l_0$-norm uncertainty set for $U_{RHS}$ as this can be represented explicitly via scenarios.

This formulation (\ref{aro_split}) can be represented as an LP, and it scales well for cost-based uncertainties in particular, meaning it is useful for large-scale, real-world problems. We prove the claim that it can be represented as an LP with the help of two lemmas.

\begin{lemma}[LP reformulation of two-stage ARO for cost-based uncertainties]\label{thm:LPCost}
    If the only uncertain parameters in a linear two-stage adaptive robust optimization model with full recourse are objective function coefficients (costs), the model can be converted to a single-stage model and solved as an LP.
\end{lemma}

Note that, later on in Theorem \ref{thm:LPSplit}, we do not formally use the conclusion of this Lemma, but rather apply similar logic as is used in this Lemma. We present the Lemma in this form, rather than a more esoteric form that could be directly applied to Theorem \ref{thm:LPSplit}, for ease of explanation.

\begin{proof}{Proof.}
    Assume we have a two-stage adaptive robust optimization formulation with uncertainties in the objective function coefficients only, as follows:
    \begin{align} \label{cost_only_formulation}
        \min_{(x_1,y_1) \in \Psi} \ \big( & c_1^T x_1 + b_1^Ty_1 + \max_{u_c \in U_c} \min_{(x_2,y_2) \in \Omega^*(x_1)} c_2(u_c)^Tx_2 + b_2(u_c)^Ty_2 \big) \\
        \text{where:} & \notag \\
        \Psi &= \{ x_1 \in \mathbb{R}^m_+, y_1 \in \mathbb{R}^n : A_1x_1+B_1y_1 \leq g_1\} \notag \\
        U_c &= \{ u_c \in \mathbb{R}^r : \| u_c \|_1 \leq \Gamma_c, \| u_c \|_{\infty} \leq 1 \} \notag \\
        \Omega^*(x_1) &= \{ x_2 \in  \mathbb{R}^m_+, y_2 \in  \mathbb{R}^n : A_2x_2+B_2y_2 \leq g_2 + Cx_1 \} \notag \\
        \Omega^*(x_1) &\neq \emptyset \  \forall \ (x_1,y_1) \in \Psi \notag
    \end{align}
    In this formulation, $U_c$ is convex and compact, and $\Omega^*(x_1)$ is a bounded polyhedral set, which combined with the assumption of full recourse means that $\Omega^*(x_1)$ is also convex and compact. Further, $c_2(u_c)^Tx_2 + b_2(u_c)^Ty_2$ is bilinear in $(u_c,(x_2,y_2))$. Therefore, by Von Neumann's Minimax theorem, we have that for fixed $x_1$, $\max_{u_c \in U_c} \min_{(x_2,y_2) \in \Omega^*(x_1)} c_2(u_c)^Tx_2 + b_2(u_c)^Ty_2 = \min_{(x_2,y_2) \in \Omega^*(x_1)} \max_{u_c \in U_c} c_2(u_c)^Tx_2 + b_2(u_c)^Ty_2$. Therefore Problem \ref{cost_only_formulation} is equivalent to the following problem:

    \begin{align} \label{max_min_flip}
        \min_{(x_1,y_1) \in \Psi} \ \big( & c_1^T x_1 + b_1^Ty_1 + \min_{(x_2,y_2) \in \Omega^*(x_1)} \max_{u_c \in U_c}  c_2(u_c)^Tx_2 + b_2(u_c)^Ty_2 \big)
    \end{align}

    Since $\min_{(x_2,y_2) \in \Omega^*(x_1)}$ is a convex function of $x_1$ (see 5.6.1 in \cite{boyd_convex_2004}),  we can further combine the two $\min$ operators to arrive at:

    \begin{align} \label{combined_mins}
        \min_{(x_1,y_1) \in \Psi, (x_2,y_2) \in \Omega^*(x_1)} \ \big( & c_1^T x_1 + b_1^Ty_1 + \max_{u_c \in U_c}  c_2(u_c)^Tx_2 + b_2(u_c)^Ty_2 \big)
    \end{align}

    Which is a single-stage robust optimization problem with uncertain objective function coefficients. Per \cite{bertsimas_price_2004}, we can move the objective function to a constraint, and separate the nominal values of $c_2$ and $b_2$ from their deviations to arrive at:

    \begin{align} \label{moved_to_constraint}
        \min & \ t \\
        \text{subject to:} & \notag \\
        & c_1^T x_1 + b_1^Ty_1 + \overline{c}_2^Tx_2 + \overline{b}_2^Ty_2 + \max_{u_c \in U_c} [ (\hat{c}_2 \cdot u_c)^Tx_2 + (\hat{b}_2 \cdot u_c)^Ty_2] \leq t \notag \\
        & (x_1,y_1) \in \Psi, (x_2,y_2) \in \Omega^*(x_1) \notag
    \end{align}

    Via strong duality as in \cite{bertsimas_price_2004}, this can be equivalently written as a linear program:

    \begin{align} \label{lp_cost_final_form}
        \min & \ t \\
        \text{subject to:} & \notag \\
        & c_1^T x_1 + b_1^Ty_1 + \overline{c}_2^Tx_2 + \overline{b}_2^Ty_2 + p \Gamma + q^T \mathbf{1} \leq t \notag \\
        & (\hat{c}_2 \cdot x_2) + (\hat{b}_2 \cdot y_2) \leq p \mathbf{1} + q \notag \\
        & p,q \geq 0 \notag \\
        & (x_1,y_1) \in \Psi, (x_2,y_2) \in \Omega^*(x_1) \notag
    \end{align}

    This concludes the proof. We conclude that linear two-stage ARO with only cost-based uncertainties represented with an $l_1$-norm budget uncertainty set are equivalent to single-stage robust optimization problems. This fact forms the foundation of the formulation presented in this paper.
\end{proof}

Note that the above proof follows similar reasoning as was introduced in \cite{arslan_decomposition-based_2022}.

\begin{lemma}[LP reformulation of two-stage ARO via scenarios]\label{thm:LPScenarios}
    Two-stage adaptive robust optimization can be equivalently reformulated via scenarios, if all scenarios in the uncertainty set are explicitly enumerated.
\end{lemma}

Note that while this reformulation is exact, the number of scenarios grows combinatorially with the number of uncertain parameters.

\begin{proof}{Proof.}
Assume we have a standard two-stage robust optimization problem as in Problem \ref{aro}, and we can enumerate the elements of $U$ as $u_i$. Then we can use the standard LP reformulation of a $\max$ operator to move the second stage objective function to a series of constraints, one for each $u_i$ (note $x_{2i}$ and $y_{2i}$ are copies of $x_2$ and $y_2$, one for each $u_i$, not subsets of $x_2$ and $y_2$):

\begin{align} \label{aro_scenarios}
    \min_{(x_1,y_1) \in \Psi} \ & c_1^T x_1 + b_1^Ty_1 + t \\
    \text{subject to:} & \notag \\
    \min_{(x_{2i},y_{2i}) \in \Omega(x_1,u_i)} & [c_2(u_i)^Tx_{2i} + b_2(u_i)^Ty_{2i}] \leq t \ \forall \ u_i \in U \notag
\end{align}

Since this is a minimization problem inside a minimization problem, we can further delete the min operator (as is common with the standard LP formulation of single-stage robust optimization) but keep its constraints to obtain:

\begin{align} \label{aro_scenarios_final}
    \min_{(x_1,y_1) \in \Psi} \ & c_1^T x_1 + b_1^Ty_1 + t \\
    \text{subject to:} & \notag \\
    & c_2(u_i)^Tx_{2i} + b_2(u_i)^Ty_{2i} \leq t \ \forall \ u_i \in U \notag \\
    & (x_{2i},y_{2i}) \in \Omega(x_1,u_i) \ \forall  \ i \notag
\end{align}

Which is a linear program. This concludes the proof.
\end{proof}

\begin{remark}
This linear program might be very large and intractable if the number of elements of $U$ is large; in other words, this reformulation does not change the fact that two-stage ARO is NP-hard. In our case, the reformulation is only useful if the number of elements of $U$ is manageably small.
\end{remark}

Armed with these two lemmas, we now present and prove the theorem that underlies our split uncertainty budget formulation:

\begin{theorem}[LP-representation of split-budget formulation]\label{thm:LPSplit}
The split-uncertainty-budget formulation presented in Problem \ref{aro_split} can be formulated as a linear program.
\end{theorem}

\begin{proof}{Proof.}
As in Lemma \ref{thm:LPScenarios}, we can first reformulate the $\max_{u_{RHS} \in U_{RHS}}$ operator with the help of scenarios:

\begin{align} \label{aro_split_step1}
    \min_{(x_1,y_1) \in \Psi} \ & c_1^T x_1 + b_1^Ty_1 + t \\
    \text{subject to:} & \notag \\
    & \max_{u_c \in U_c} \min_{(x_{2i},y_{2i}) \in \tilde{\Omega}(x_1,u_i)} [c_2(u_c)^Tx_{2i} + b_2(u_c)^Ty_{2i}] \leq t \ \forall \ u_i \in U_{RHS} \notag
\end{align}

We can then flip the max and the min as in Lemma \ref{thm:LPCost}, and drop the min operator, to arrive at:

\begin{align} \label{aro_split_step2}
    \min_{(x_1,y_1) \in \Psi} \ & c_1^T x_1 + b_1^Ty_1 + t \\
    \text{subject to:} & \notag \\
    & \max_{u_c \in U_c}  [c_2(u_c)^Tx_{2i} + b_2(u_c)^Ty_{2i}] \leq t \ \forall \ u_i \in U_{RHS} \notag \\
    & (x_{2i},y_{2i}) \in \tilde{\Omega}(x_1,u_i) \ \forall \ i \notag
\end{align}

This is now a linear program with uncertain constraint coefficients, which as in Lemma \ref{thm:LPCost} can be reformulated as:

\begin{align} \label{aro_split_step3}
    \min_{(x_1,y_1) \in \Psi} \ & c_1^T x_1 + b_1^Ty_1 + t \\
    \text{subject to:} & \notag \\
    & \overline{c}_2^T x_{2i} + \overline{b}_2^T y_{2i} + p_i \Gamma_c + q_i^T \mathbf{1} \leq t \ \forall \ i \notag \\
    & (\hat{c}_2 \cdot x_{2i}) + (\hat{b}_2 \cdot y_{2i}) \leq p_i \mathbf{1} + q_i  \ \forall \ i\notag \\
    & (x_{2i},y_{2i}) \in \tilde{\Omega}(x_1,u_i) \ \forall \ i \notag \\
    & p,q \geq 0 \notag
\end{align}

This concludes the proof.

\end{proof}

\begin{remark}
    An intuitive way to interpret this formulation is that it is an LP capacity expansion planning problem featuring uncertain objective function cost coefficients and scenarios representing additional second-stage uncertainties in right-hand-side constraint parameters (such as constraints on resource availability or maximum deployable capacity), with the maximum of those scenarios entering the objective function.
\end{remark}

\subsection{Advantages}\label{subsec:Advantages}

The primary advantage of this formulation is that it is highly scalable for cost uncertainties in particular. As opposed to standard two-stage ARO, which is tractable for many problems but not scalable enough to use for large-scale capacity expansion, many cost-based uncertainties can be modeled in this formulation, with minimal impact on computational times, as we will later see in the case study. In this manner, the formulation leverages partial convexities of two-stage ARO that specifically apply to cost-based uncertainties, thus enabling scalability for real-world problem instances.

It is important to note, however, that this formulation scales combinatorially with regard to the number of non-cost scenarios (e.g. uncertainties in right-hand-side values) modeled, making it still NP-hard with respect to non-cost uncertainties. However, as we will demonstrate in the case study, this is not a limitation for many real-world problems, as long as decision-makers are concerned with only a few ($\sim 5-10$) non-cost, non-weather-based uncertainties.

\subsection{A note on unavailability as high cost}\label{subsec:CostNote}

One question that naturally arises when considering this formulation is: if cost-based uncertainties can be dealt with cleanly but right-hand-side-based uncertainties cannot, why not convert right-hand-side uncertainties to cost uncertainties using objective function penalties? For example, consider a two-stage ARO problem with uncertainty only in the availability of second-stage investment decisions:

\begin{align} \label{aro_rhsonly}
    \min_{(x_1,y_1) \in \Psi} \ \big( & c_1^T x_1 + b_1^Ty_1 + \max_{u \in U} \min_{(x_2,y_2) \in \Omega(x_1,u)} c_2^Tx_2 + b_2^Ty_2 \big) \\
    \text{where:} & \notag \\
    \Psi &= \{ x_1 \in \mathbb{R}^m_+, y_1 \in \mathbb{R}^n : A_1x_1+B_1y_1 \leq g_1\} \notag \\
    U &= \{ u \in \mathbb{R}^p : \| u \|_0 \leq \Gamma, 0 \leq \| u \|_{\infty} \leq 1 \} \notag \\
    \Omega(x_1, u) &= \{ x_2 \in  \mathbb{R}^m_+, y_2 \in  \mathbb{R}^n : A_2x_2+B_2y_2 \leq g_2 + Cx_1, x_2 \leq \overline{x}_2(1-u) \} \notag \\
    \Omega(x_1, u) &\neq \emptyset \  \forall \ (x_1,y_1) \in \Psi, u \in U \notag
\end{align}

Where $\overline{x}_2$ is the nominal maximum availability for $x_2$. Intuitively, this is a two stage adaptive robust capacity expansion problem where $\Gamma$ of the second-stage investments are assumed to be unavailable in each worst case within the uncertainty set. One can easily form an equivalent problem with objective function uncertainties only by using a penalty term $M$ to move the constraint $x_2 \leq \overline{x}_{2}(1-u)$ to the objective:

\begin{align} \label{aro_rhsonly_penaltyform}
    \min_{(x_1,y_1) \in \Psi} \ \big( & c_1^T x_1 + b_1^Ty_1 + \max_{u \in U} \min_{(x_2,y_2) \in \Omega(x_1)} c_2^Tx_2 + b_2^Ty_2 + Mu^Tx_2 \big) \\
    \text{where:} & \notag \\
    \Psi &= \{ x_1 \in \mathbb{R}^m_+, y_1 \in \mathbb{R}^n : A_1x_1+B_1y_1 \leq g_1\} \notag \\
    U &= \{ u \in \mathbb{R}^p : \| u \|_0 \leq \Gamma, 0 \leq \| u \|_{\infty} \leq 1 \} \notag \\
    \Omega(x_1) &= \{ x_2 \in  \mathbb{R}^m_+, y_2 \in  \mathbb{R}^n : A_2x_2+B_2y_2 \leq g_2 + Cx_1, x_2 \leq \overline{x}_2 \} \notag \\
    \Omega(x_1) &\neq \emptyset \  \forall \ (x_1,y_1) \in \Psi \notag \\
    M &>> 1 \notag
\end{align}

Where M is a large scalar sufficiently high to enforce $u^Tx_2=0$. This ensures the complementarity of $u$ and $x_2$: if $u_i=1$, the corresponding $x_{2i}$ will be equal to zero, and vice versa.

The above formulation now has objective function uncertainties only, but it is still not convex due to the cardinality constraint on $u$, $\| u \|_0 \leq \Gamma$. We can use the $l_1$-norm convex relaxation of this constraint, as is standard with two-stage ARO, to obtain the following convex problem:

\begin{align} \label{aro_rhsonly_penaltyform_convex}
    \min_{(x_1,y_1) \in \Psi} \ \big( & c_1^T x_1 + b_1^Ty_1 + \max_{u \in U_{cvx}} \min_{(x_2,y_2) \in \Omega(x_1)} c_2^Tx_2 + b_2^Ty_2 + Mu^Tx_2 \big) \\
    \text{where:} & \notag \\
    \Psi &= \{ x_1 \in \mathbb{R}^m_+, y_1 \in \mathbb{R}^n : A_1x_1+B_1y_1 \leq g_1\} \notag \\
    U_{cvx} &= \{ u \in \mathbb{R}^p : \| u \|_1 \leq \Gamma, \| u \|_{\infty} \leq 1 \} \notag \\
    \Omega(x_1) &= \{ x_2 \in  \mathbb{R}^m_+, y_2 \in  \mathbb{R}^n : A_2x_2+B_2y_2 \leq g_2 + Cx_1, x_2 \leq \overline{x}_2 \} \notag \\
    \Omega(x_1) &\neq \emptyset \  \forall \ (x_1,y_1) \in \Psi \notag \\
    M &>> 1 \notag
\end{align}

Which as per Lemma \ref{thm:LPCost} can be equivalently formulated as a single-stage LP. Intuitively, this formulation appears to be equivalent to the original problem (\ref{aro_rhsonly}), as it represents using cost penalties to model unavailability of second-stage investments. However, upon further inspection, it is actually a relaxation that is unlikely to be tight in general, meaning it is an upper bound that may be overly conservative. Theorem \ref{thm:relaxation} shows this:

\begin{theorem}[Convex relaxation of two-stage ARO]\label{thm:relaxation}
Problem (\ref{aro_rhsonly_penaltyform_convex}), the convex relaxation of (\ref{aro_rhsonly_penaltyform}), is not tight for large M under some mild assumptions that are likely to be met in real-world applications.
\end{theorem}

\begin{proof}{Proof.}

Assume we have a solution $(x_1^*,y_1^*,u^*,x_2^*,y_2^*)$ that is optimal to (\ref{aro_rhsonly}), and by extension (\ref{aro_rhsonly_penaltyform}) if M is sufficiently large to enforce $u^{*T}x_2^*=0$. Assume M is large enough such that if it were to be decreased slightly, the solution to (\ref{aro_rhsonly_penaltyform}) would not change. Also assume that there is at least one nonzero element of $x_2^*$ at optimality (this is the mild assumption referenced above); we pick one such element and refer to it as $x_{2j}$. Due to the structure of the budget uncertainty set, at optimality we will have $\Gamma$ elements of $u^*$ equal to 1 and $p-\Gamma$ elements equal to 0. Now, let $\tilde{u}$ be the same as $u^*$ except with one of its elements $u_i$ that is equal to 1 decreased slightly, and another element $u_j$ that is equal to 0, corresponding to the nonzero $x_{2j}$, increased slightly. This less-sparse $\tilde{u}$ is within the set uncertainty $U_{cvx}$ used in (\ref{aro_rhsonly_penaltyform_convex}) but not in the uncertainty set $U$ used in (\ref{aro_rhsonly_penaltyform}) and (\ref{aro_rhsonly}). If we substitute $\tilde{u}$ for $u^*$ in (\ref{aro_rhsonly_penaltyform_convex}), the objective function will become higher, because $u_j*x_{2j}>0$, before any adaptations by the inner minimization operator. In response to this new $u=\tilde{u}$, the inner minimization will keep all of the elements of $x_2$ that were previously 0 in $x_2^*$ the same, because M is sufficiently high to force this, and will either: a) keep $x_{2j}$ the same, thus arriving at the same solution set $(x_1^*,y_1^*,x_2^*,y_2^*)$, but with a higher objective or b) reduce $x_{2j}$ to zero to avoid the new penalty induced by $u_j*x_{2j}>0$, which can only increase the objective. Therefore, given an optimal point to (\ref{aro_rhsonly_penaltyform}), we can construct a feasible point to (\ref{aro_rhsonly_penaltyform_convex}) with a higher objective function value, meaning the relaxation in (\ref{aro_rhsonly_penaltyform_convex}) is not tight.
\end{proof}

\begin{remark}
It is possible to formulate this relaxation more generally as: \begin{align} \label{aro_rhs_general}
    \min_{(x_1,y_1) \in \Psi} \ \big( & c_1^T x_1 + b_1^Ty_1 + \max_{u \in U_{cvx}} \min_{(x_2,y_2,z) \in \Omega(x_1)} c_2^Tx_2 + b_2^Ty_2 + Mu^Tz \big) \\
    \text{where:} & \notag \\
    \Psi &= \{ x_1 \in \mathbb{R}^m_+, y_1 \in \mathbb{R}^n : A_1x_1+B_1y_1 \leq g_1\} \notag \\
    U_{cvx} &= \{ u \in \mathbb{R}^p : \| u \|_1 \leq \Gamma, \| u \|_{\infty} \leq 1 \} \notag \\
    \Omega(x_1) &= \{ x_2 \in  \mathbb{R}^m_+, y_2 \in  \mathbb{R}^n, z \in \mathbb{R}^p : A_2x_2+B_2y_2-Cx_1 \leq (\overline{g}_2-\hat{g}_2)+\hat{g}_2 \cdot z \} \notag \\
    \Omega(x_1) &\neq \emptyset \  \forall \ (x_1,y_1) \in \Psi \notag \\
    M &>> 1 \notag
\end{align}

Where $z$ is an auxiliary variable in the same space as $u$, $\overline{g}_2$ is the nominal value vector of $g_2$, and $\hat{g}_2$ is the vector of maximum deviations from $\overline{g}_2$. We choose to present the relaxation in the more restricted and simplified form (\ref{aro_rhsonly_penaltyform_convex}) for ease of explanation.
\end{remark}

The conclusion of Theorem \ref{thm:relaxation} is of course not surprising given that two-stage ARO is an NP-hard formulation. However, this conclusion does not mean that this relaxation is not useful. With properly tuned $M$, it could be a useful relaxation, similar to other relaxations with tunable parameters. The extent to which it is a useful relaxation with properly tuned $M$ is an important area for future research. This observation that it is unlikely to be tight, however, helps justify the approach taken in the split-budget formulation in this paper, where we separate out cost and non-cost uncertainties, rather than attempt to reformulate some non-cost uncertainties (for example, availability constraints) as cost uncertainties. In our split budget formulation, $u_c$ is not guaranteed to have $\|u_c\|\leq \Gamma$, due to the same underlying factors driving the proof of Theorem \ref{thm:relaxation}, but this is not problematic as long as our uncertainty set represents true cost uncertainties (meaning a more sparse solution would still be within the desired uncertainty set), rather than cost penalties that are being used to parametrize availability uncertainty.

\section{Real-world case study}\label{sec:CaseStudy}

In this section we describe a case study of our formulation, in which it is applied to a real-world, large-scale capacity expansion problem for the state of California. We use the open-source GenX model \citep{jenkins_enhanced_2017, bonaldo_genx_2024} and develop an instance of the model that is benchmarked to replicate the RESOLVE model used by the California Public Utilities Commission for its Integrated Resource Planning process, which informs generation and transmission planning for the State and its load-serving entities \citep{california_public_utilities_commission_inputs_2023}. Our case study has two parts: the first part demonstrates the impact of adding robustness to the standard deterministic model, without our split-budget formulation, and the second part compares our split-budget formulation to the standard formulation and evaluates its performance on both outcomes and computational intensity.

\subsection{Deterministic model development and benchmarking}\label{subsec:ModelDev}

To establish our case study, we use the publicly available RESOLVE model data directly, meaning the deterministic version of our instance of GenX has exactly the same input data as the RESOLVE model. We make a slight adjustment to the model for use in our case study, in which we remove many of the near term resource builds that are ``forced-in'' (by exogenous constraint) in the CPUC IRP version of RESOLVE. In our case study, we only include forced-in builds corresponding to resources under contract by load-serving entities. We make this adjustment because many of these ``forced-in'' resources are far from certain to be built, meaning they do not necessarily belong in a least-cost optimization, and because the volume of forced-in resources is so significant that they over-constrain the model, making it unlikely for structural model changes, such as the ones introduced by this study, to substantially impact near-term capacity decisions. Figure \ref{fig:benchmarking} demonstrates the benchmarking of our GenX case to RESOLVE, and shows that resource builds in the snapshot 2035 year are nearly exactly the same between the two models. It also shows the impact of adjusting the aforementioned minimum build constraints; the main result is that approximately 5 gigawatts (GW) less of battery storage is included when these constraints are modified.

\begin{figure}
     \FIGURE
    {\includegraphics[width=\textwidth]{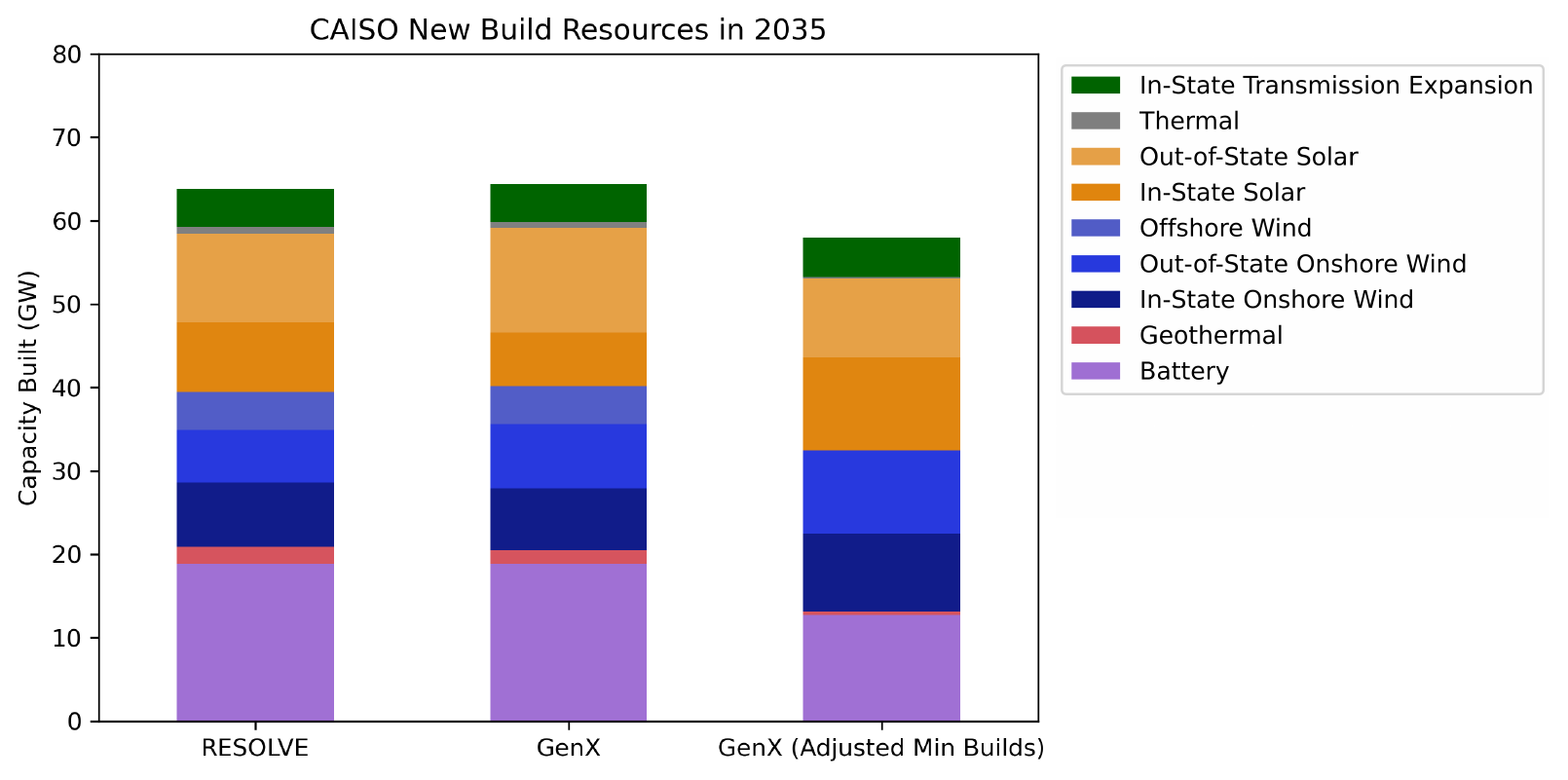}}
{Benchmarking Results: CAISO New Build Resources in 2035. \label{fig:benchmarking}}
{CAISO = California Independent System Operator.}
\end{figure}

A further adjustment that is made to our model is to reduce the number of planning periods from 12 to 2. We use planning periods of 2030, which represents years 2026-2030, and 2045, which represents years 2031-2045. This change is made for two reasons: a) to increase tractability, allowing for a scenario-based robust model that we use to compare to our formulation, the former of which would not be tractable for the full set of planning periods, and b) to cleanly separate build decisions into two stages: the ``here-and-now'' decisions that represent decisions being made in the near term by the resource planning process that the model is being used for, and the ``wait-and-see'' decisions that can be adapted later in response to new information and the resolution of key uncertainties. This two-stage framing is both necessary for two-stage adaptive robust optimization, and more representative of the decision context that the model is used for. In addition to reducing the number of planning periods, we also make the change of forcing all transmission expansion decisions to be made in the first stage, which is reflective of the fact that transmission lines in California can take over 10 years to build. Therefore, in our model, the first-stage decisions consist of 2026-2030 generation and storage capacity decisions and operations, as well as 2026-2045 transmission expansion decisions, while the second-stage decisions consist of 2031-2045 generation and storage capacity decisions and operations. The end result is a model with about 6 million variables and 3 million constraints, solved with the Gurobi commercial solver on a high performance computing cluster.

\subsection{Experimental design}\label{subsec:ExpDesign}

In our case study, we start from our benchmarked, two-stage deterministic GenX model, then explore the impact of adding robustness to the model. Our case study has two goals: 1) to examine the impact of applying a standard adaptive robust optimization framework to the deterministic model, which provides a baseline for comparison to our split-budget formulation, and 2) to examine the performance of our split-budget formulation in comparison to the standard formulation. 

\subsubsection{Standard adaptive robust optimization model}

To examine the impact of adding robustness to the deterministic model, we use a scenario-based two-stage adaptive robust optimization formulation as described above in Lemma \ref{thm:LPScenarios}. (This scenario-based approach in and of itself is not novel and is similar to the approach taken by \cite{moreira_climateaware_2021}.) We model seven key uncertainties, which we developed in coordination with stakeholders to reflect the most important uncertainties impacting the California PUC's resource planning process, which are enumerated in Table \ref{tab:uncertain-parameters}. The uncertainty ranges were chosen to represent a range of future trajectories that is plausible for each parameter.

\begin{table}[h!]
\centering
\begin{tabular}{p{0.55\textwidth} p{0.35\textwidth}}
\hline
\textbf{Uncertain Parameter} & \textbf{Uncertainty Range} \\
\hline
Solar, wind, and storage costs & +0 to +50\% \\
Natural gas prices &  +0 to +100\% \\
Fixed cost of maintaining thermal fleet &  +0 to +100\% \\
Load increases due to data centers and electrification &  +0 to +10\% above baseline projections\\
Out-of-state renewables availability & Limited availability: 2 GW wind, 5 GW solar \\
Offshore wind availability & Limited availability: no offshore wind available \\
Availability of imports during peak hours & -0 to -50\% \\
\hline
\end{tabular}
\caption{Key uncertain parameters and assumed ranges}
\label{tab:uncertain-parameters}
\end{table}

In the adaptive robust formulation, worst-case scenarios consist of combinations of these uncertain parameters taking their worst-case values. We use the parameter $\Gamma$, the uncertainty ``budget,'' to control the number of uncertain parameters taking their worst case values in each scenario. Since this scenario-based formulation enumerates all possible scenarios within the uncertainty set, this means that the model size scales combinatorially with the number of scenarios, making it intractable for large numbers of uncertainties. However as we demonstrate in our case study, the tractability is greatly increased when our split-budget formulation is used.

The first step in our case study is to examine the change to first-stage resource builds when this standard robust formulation is used, for a range of values of $\Gamma$ between 1 (resulting in 7 second-stage scenarios, each reflecting one parameter taking on a worst-case outcome) and 3 (resulting in 35 second-stage scenarios, reflecting all combinations of 3 possible worst-case outcomes). To determine the benefit of this added robustness, we also perform a second, stress-testing step, in which portfolios are stress-tested against a range of possible futures. (This stress-testing step is a crucial part of a strong uncertainty analysis framework regardless of whether advanced robust or stochastic formulations are used, as discussed in \cite{mantegna_integrated_2026}.) For the stress-testing step, we use the same uncertainties as previously discussed, as well as four additional uncertainties that primarily represent ``upsides'', meaning system costs could turn out better than expected if they deviate from their nominal values, described in Table \ref{tab:upside-parameters}.

\begin{table}[h!]
\centering
\begin{tabular}{p{0.55\textwidth} p{0.35\textwidth}}
\hline
\textbf{Uncertain Parameter} & \textbf{Uncertainty Range} \\
\hline
Electric vehicle load flexibility & 0 GW (baseline) to 10 GW available \\
Enhanced geothermal (EGS) availability & Baseline geothermal availability, to double baseline geothermal availability with 50\% lower cost \\
Diablo Canyon nuclear power plant extension &  Retire after 2030 (baseline) / Extended to 2045 \\
Import and export availability &  +0 to +50\% above baseline projections\\
\hline
\end{tabular}
\caption{Key ``upside'' uncertain parameters and assumed ranges}
\label{tab:upside-parameters}
\end{table}

These upside scenario tests are meant to gauge whether a robust strategy may lose out on opportunities for cost savings should favorable outcomes occur, or whether the strategy remains flexible enough to take advantage of these occurrences.

In our stress testing, we stress test first-stage portfolios across a range of realizations for these 11 total uncertain parameters. In the stress testing results, we show the results for scenarios where two uncertain parameters at a time deviate from their nominal values (to take their worst-case value in the case of the downside uncertainties, or their best-case value in the case of the upside uncertainties), which is a somewhat arbitrary choice, but one that also highlights the advantages of a ``middle of the road'' level of robustness where $\Gamma=2$.

\subsubsection{Split-budget formulation}

Following the examination of the impact of the standard robust formulation on both first-stage capacity builds and second-stage performance, we examine the impact of our split-budget formulation. While the split-budget formulation could handle a more granular set of cost uncertainties than the three listed above, we choose to keep the uncertainty set for costs the exact same as in the standard formulation so as to achieve an apples-to-apples comparison of the performance of our formulation. Therefore when we model the split uncertainty budget, our cost uncertainty set $U_c$ includes three cost parameters that are allowed to vary from their nominal values (items 1-3 in the 7-item list above), and our scenario-based uncertainty set $U_{RHS}$ includes four right-hand-side parameters that are allowed to vary from their nominal values (items 4-7 in the 7-item list above). In this portion of the case study, we compare the split-budget formulation to the standard formulation on the basis of first-stage builds, second-stage performance under each 'stress test' scenario, and solver times.

\subsection{Standard robust model results}\label{subsec:MainResults}

Figure \ref{fig:results_main} shows the results of comparing first-stage capacity builds in the deterministic model to those in the standard two-stage robust model, for varying levels of robustness from $\Gamma=1$ to $\Gamma=3$. The results indicate that the robust model chooses to invest significantly more in in-state transmission as compared to the deterministic model, with substantially greater in-state transmission expansion when $\Gamma\ge2$. The results for $\Gamma=3$ show slightly smaller transmission capacity expansion in MW compared to $\Gamma=2$, but in reality, transmission investments are shifted to more expensive projects that offer a wider range of benefits but are more expensive per megawatt of transfer capacity, meaning total transmission investment in dollars is greatest when $\Gamma=3$. This dynamic is illustrated in Figure \ref{fig:tx_build_maps_panel}, which shows the centroids of areas that are affected by transmission upgrades selected by the model, along with the size of the upgrades (in the RESOLVE model as in our benchmarked GenX model, transmission upgrades are represented via generic upgrade costs to increase deliverability for certain areas, with costing provided by California ISO studies, rather than specific candidate transmission paths that can be mapped as linear infrastructure). The robust model shows increased investment in transmission lines that provide adaptivity to ensure deliverability of resources that might be needed under possible futures, but would not be optimal in a deterministic case in which the future is assumed to be known with perfect certainty. Specifically, for $\Gamma=2$ and $\Gamma=3$, the model chooses to invest in transmission upgrades to deliver Southern Central Valley resources, and in the case of $\Gamma=3$, the model invests in a small amount of transmission to deliver offshore wind on the Humboldt Coast in the northern part of the State. The value of these lines is illustrated in the stress test cost results below (see Figure \ref{fig:second_stage_cost_distribution}).

\begin{figure}
     \FIGURE
    {\includegraphics[width=\textwidth]{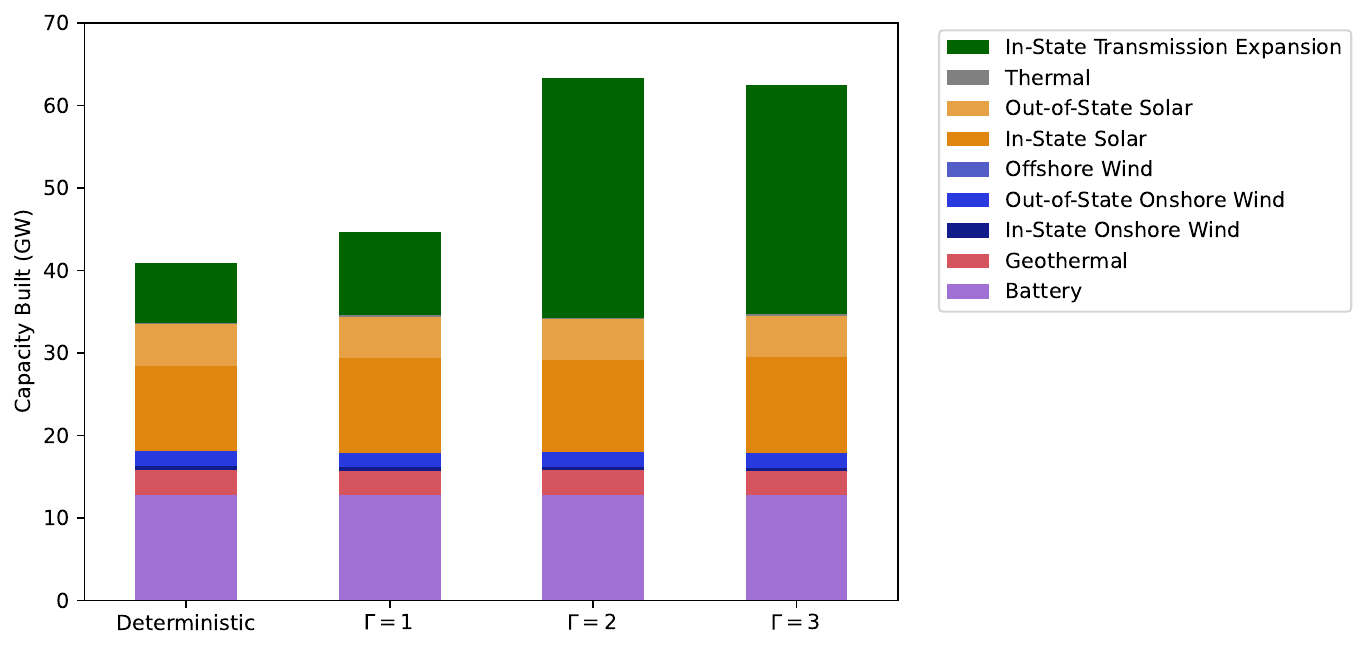}} 
{First-stage resource builds in deterministic vs robust models, with varying levels of robustness. \label{fig:results_main}}
{$\Gamma$ refers to the number of uncertain parameters that are allowed to deviate to their worst-case values.}
\end{figure}

\begin{figure}
     \FIGURE
    {\includegraphics[width=\textwidth]{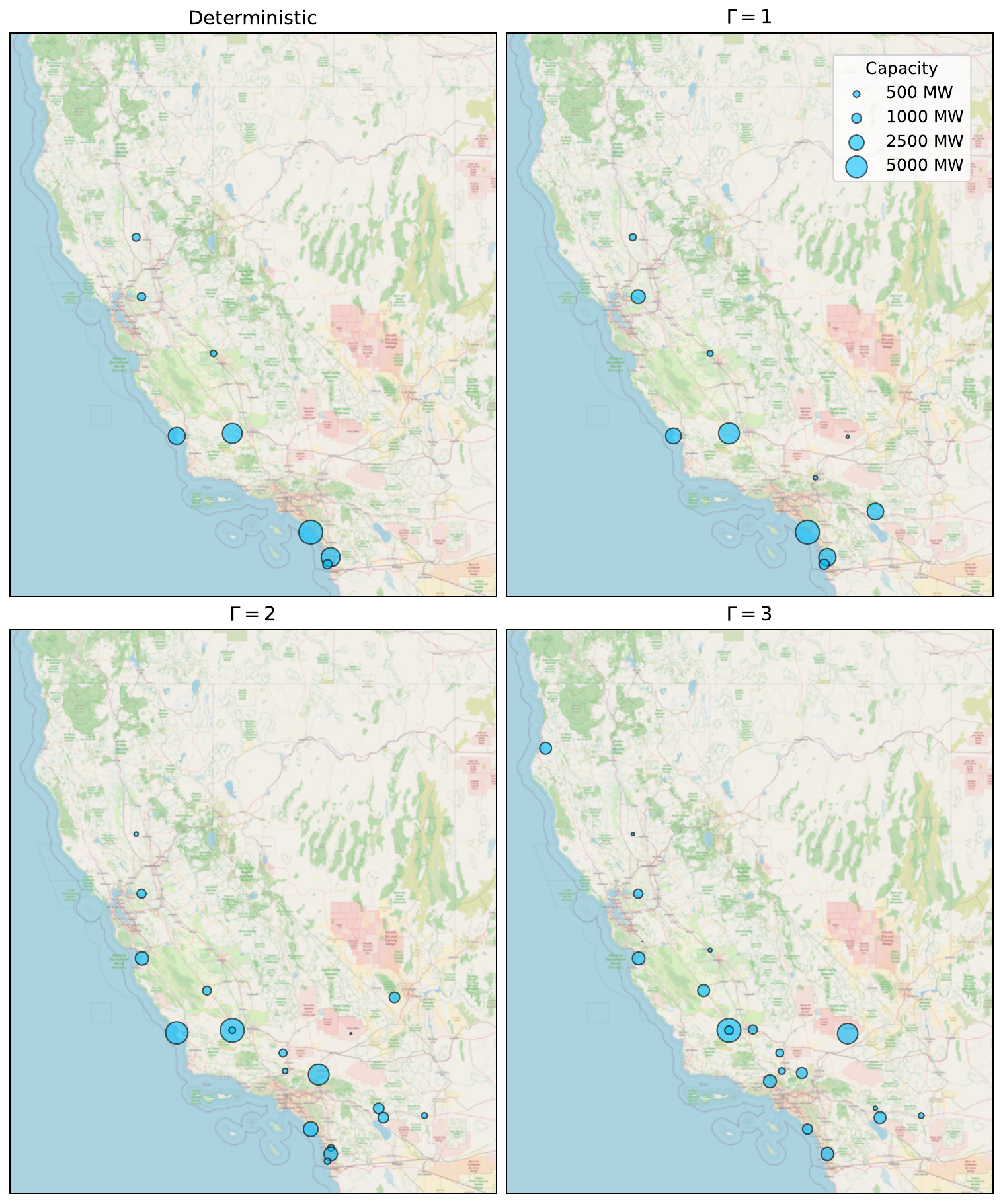}} 
{Spatial distribution of transmission builds in the four portfolios. \label{fig:tx_build_maps_panel}}
{Bubbles show the location of the centroids of the geographical areas where deliverability improvements (from the designated regions to load regions) are selected by the model. These decisions correspond to generic costs of deliverability improvements for specific regions and not to specific transmission paths.}
\end{figure}

Figure \ref{fig:second_stage_cost_distribution} shows the impact of added robustness on both near-term, first stage costs (left panel), and the range of long-term cost outcomes under each second-stage stress test scenario (right panel). The results demonstrate that adding robustness increases near term costs slightly, between \$0-1B/yr (0-0.4 cents/kWh) depending on the value of $\Gamma$, while potentially saving more than \$20B/yr (8 cents/kWh) in long-term ratepayer costs in certain worst-case scenarios as shown by the tail end of the box-and-whisker plots. Notably, for $\Gamma=1$, near-term cost increases are negligible, while the small amount of added transmission built in this case saves over \$10B/yr (4 cents/kWh) in the worst-case second stage outcome. (These figures include non-served energy costs, which are not costs that directly translate to rates; however, if non-served energy costs are not included as shown in Figure \ref{fig:second_stage_cost_distribution}, potential second-stage cost savings still are about \$5B/yr or 2 cents/kWh.) This demonstrates the ability for the robust model, which endogenizes uncertainty, to identify portfolios that offer significant protection against downside risks, with negligible upfront cost increases. On the other hand, the deterministic model, which assumes the future is known with perfect certainty, does not take advantage of this ability to significantly increase robustness with only very small changes. This highlights an important dynamic commonly found in uncertainty-aware modeling: often there is a significant ``near-optimal'' space of decisions that are near the deterministic optimum in terms of cost and that offer significant protection against downside risks, but that would never be selected by a deterministic least-cost optimization model. This is especially likely for models with a very large decision space such as the one used in this study; for decision environments in which there is a narrow range of possible decisions, uncertainty-aware modeling is less useful and a simple scenario analysis framework is likely to suffice \citep{mantegna_integrated_2026}.

\begin{figure}
     \FIGURE
    {\includegraphics[width=\textwidth]{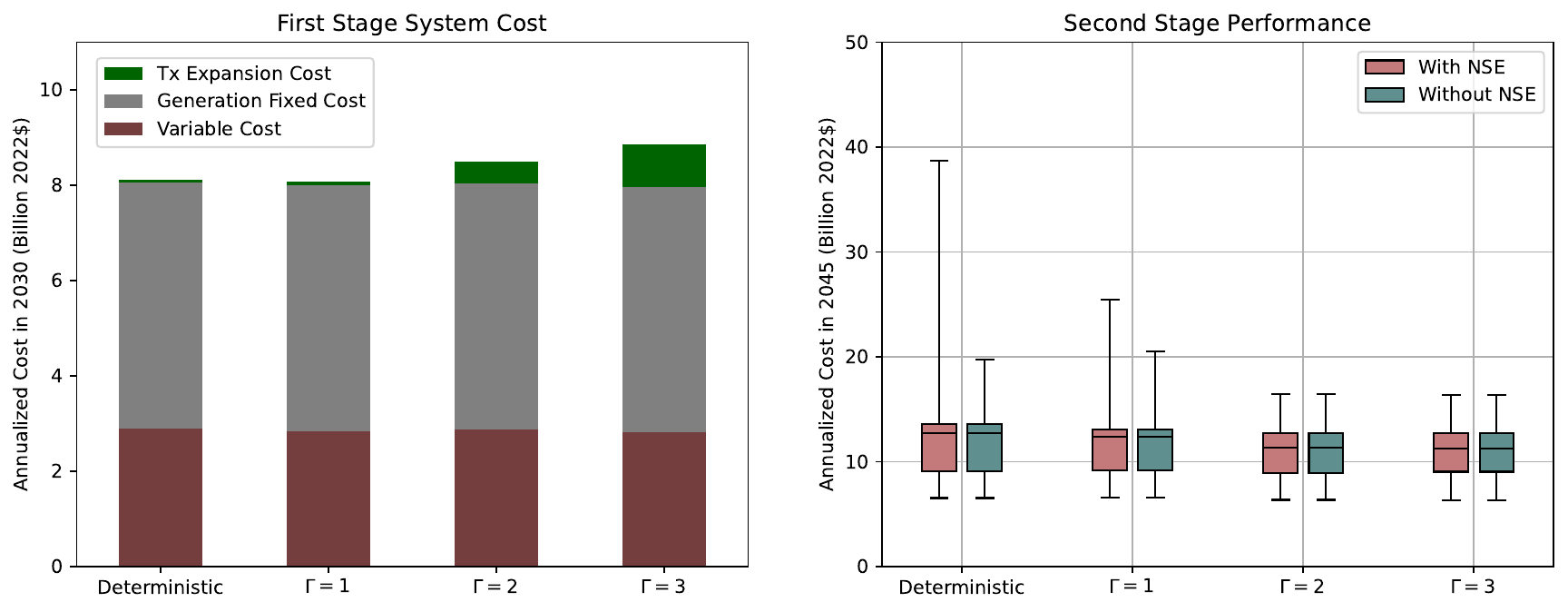}} 
{First stage cost increase by scenario, and second stage performance, from stress-testing. \label{fig:second_stage_cost_distribution}}
{Whiskers show absolute minimum and maximum scenarios.}
\end{figure}

It is also possible to identify the ``worst case'' that drives particular robust solutions, in order to get some insight into the results, particularly in the scenario-based framework used here. The worst-case modeled in the $\Gamma=2$ case, as an illustrative example, is the second stage scenario in which load is higher than expected due to a combination of electrification and data centers and out-of-state renewables availability is lower than expected; the combination of these factors means that more in-state renewables are needed in the second stage, for which there must be transmission available if the resources are to be deliverable to load. The risk-aware robust model for this $\Gamma=2$ case decides in stage one to build more in-state transmission to hedge against this possibility, as evidenced in Figures \ref{fig:results_main} and \ref{fig:tx_build_maps_panel}, and as a result, first stage costs increase by about \$0.5B/yr. However, this same worst-case is also what drives the high end of the second stage costs in the right panel of Figure \ref{fig:second_stage_cost_distribution}. Thus we see that this small amount of increased first-stage transmission investment leads to over \$20B/yr of savings in the second stage compared to the deterministic case, mainly due to high levels of non-served energy since the model cannot build enough deliverable resources to serve the unexpectedly high load, or about \$5B/yr of savings if non-served energy costs are not included.

Figure \ref{fig:second_stage_cost_upside} shows the performance of the deterministic portfolios compared to the robust portfolios in both the ``deterministic'' second-stage scenario where all uncertain parameters take their nominal values, and the four ``upside'' scenarios in which the 4 upside parameters respectively take their best-case values. The results indicate that the robust portfolios do not lead to lower performance in these deterministic and upside scenarios; rather, the robust portfolios actually sometimes offer slightly better performance in these scenarios, although the difference is minimal. (One may ask why some robust portfolios do better in the deterministic case; the reason is that the deterministic model is optimizing for both first- and second-stage costs, rather than only second-stage costs which are shown here.) One notable difference, although similarly minor, is that while the deterministic case does not show any decreased second-stage costs in the ``EGS'' scenario in which EGS becomes available at low cost, the robust portfolios do show a cost savings in this case, due to the increased investment in transmission which allows for this resource to be deliverable if and when it becomes available. Thus the results indicate that in this case study, transmission investments are highly adaptive, becoming useful both in downside cases as well as upside cases.

\begin{figure}
     \FIGURE
    {\includegraphics[width=\textwidth]{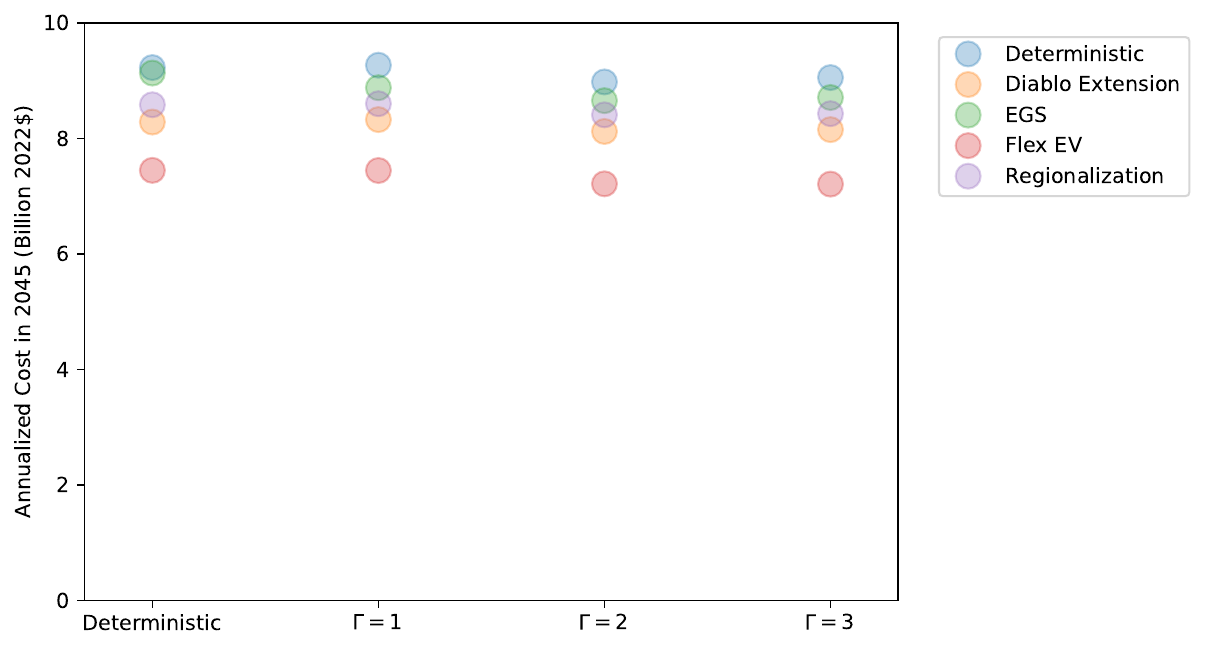}} 
{Performance of portfolios in deterministic and upside scenarios. \label{fig:second_stage_cost_upside}}
{}
\end{figure}

\subsection{Split uncertainty budget results}\label{subsec:ComparisonResults}

Figure \ref{fig:formulation_comparison_main} shows a comparison between capacity builds for the standard robust portfolios and the portfolios produced with the split budget formulation described in Section \ref{sec:SplitBudgetFormulation}. The split budget formulations are shown next to their standard combined-budget counterparts: for $\Gamma=2$ there is only one corresponding split-budget portfolio, $\Gamma_{RHS}=1, \Gamma_{cost}=1$, while for $\Gamma=3$ there are two such combinations. The results show that the split-budget portfolios are similar to their standard counterparts, although generally have a lower level of transmission build and therefore robustness. This is because, as discussed above, the split budget models a subset of the full, combined uncertainty set. The results also highlight that different allocations of uncertainty budget to cost vs right-hand-side, for the same total uncertainty budget, can impact model outcomes.

\begin{figure}
     \FIGURE
    {\includegraphics[width=\textwidth]{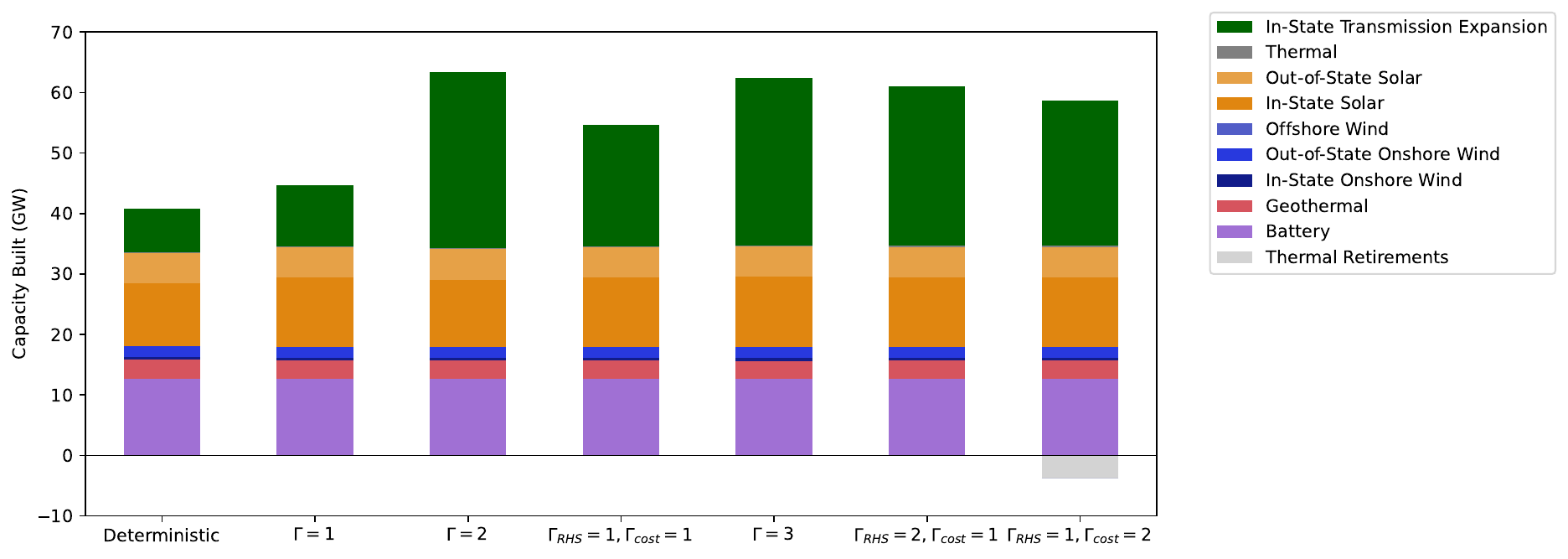}} 
{Comparison of first-stage builds between standard and split-budget robust formulations. \label{fig:formulation_comparison_main}}
{}
\end{figure}

Figure \ref{fig:formulation_comparison_costs} shows the impact on first-stage cost and second-stage robustness from using the split budget formulation. The results show that the changes to both first-stage cost and second-stage performance are largely proportional to the change in transmission builds: for example the $\Gamma_{RHS}=1, \Gamma_{cost}=1$ scenario has a lower first-stage transmission build compared to its $\Gamma=2$ counterpart, accompanied by a proportionally lower first-stage cost, as well as slightly worse second-stage performance. The other split-budget portfolios show a similar dynamic when compared to their combined-budget counterparts. In short, the split-budget formulations offer a decreased level of robustness compared to the combined budget formulation due to modeling only a subset of the full uncertainty set, with slightly decreased transmission builds along with predictable changes to both first-stage cost and second-stage performance. That said, all split-budget models result in substantially improved robustness relative to a deterministic case, offering substantial 'downside protection' against worst-case outcomes and a similar range of potential costs under stress test scenarios as the standard robust formulation.

\begin{figure}
     \FIGURE
    {\includegraphics[width=\textwidth]{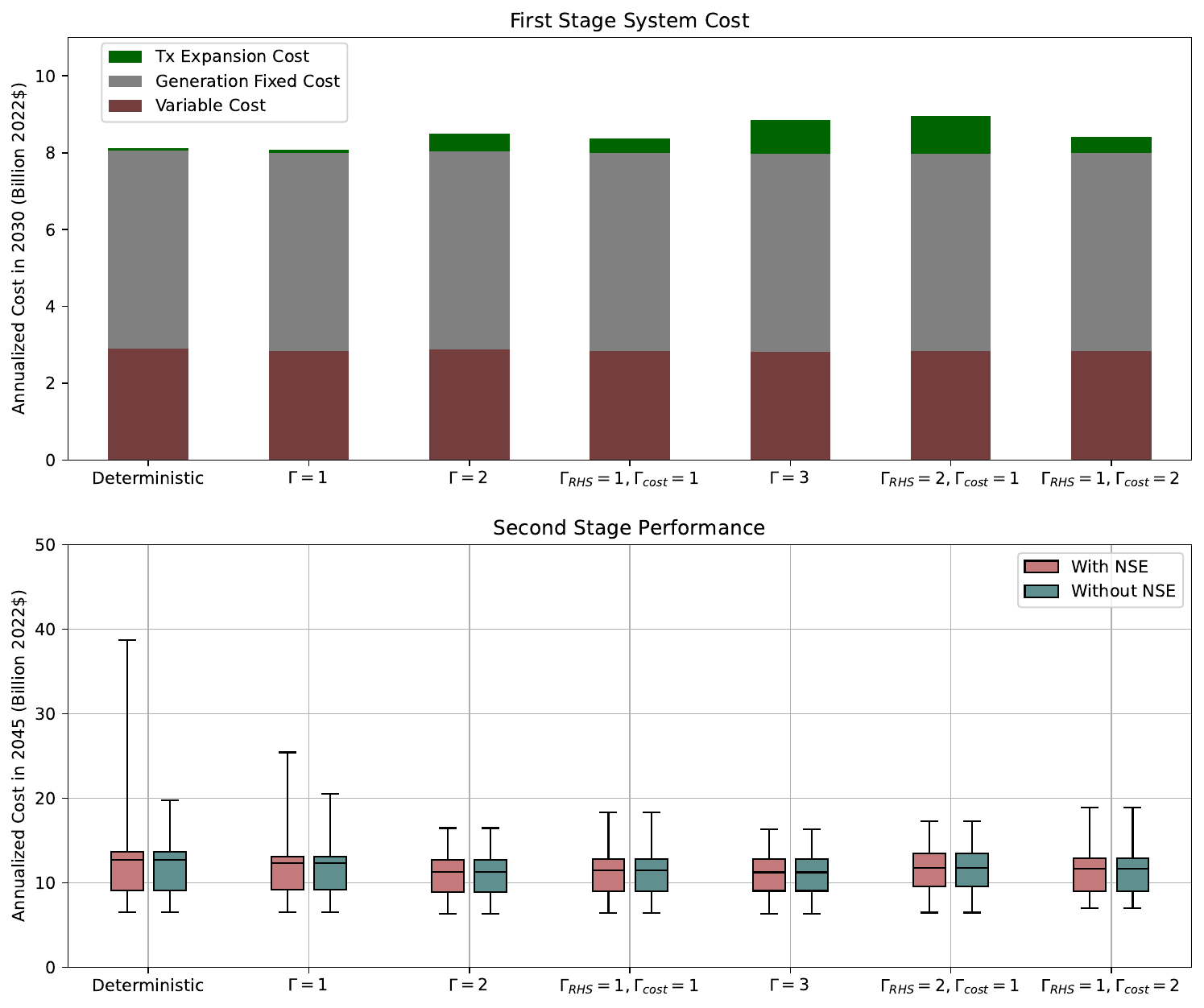}} 
{Cost results comparison between standard and split-budget robust formulations. \label{fig:formulation_comparison_costs}}
{}
\end{figure}

Figure \ref{fig:solver_time} shows that although the split-budget portfolios offer slightly less robustness compared to their full-budget counterparts, they offer significant advantages in terms of solution times. This difference is especially significant for the $\Gamma=3$ case, in which the corresponding split-budget cases have solution times that are about an order of magnitude lower. Due to the combinatorial scaling of the scenario-based robust model, this difference would be even more pronounced if more sources of uncertainty were modeled. This computational advantage also means that the split-budget formulation could be used to model cases that would be intractable using the standard formulation; for example cases in which more planning periods were included. Notably, all three split-budget cases show lower solution times than even the $\Gamma=1$ case, since the split-budget cases only model scenario-based uncertainties for 4 out of 7 of the total uncertainties.

\begin{figure}
     \FIGURE
    {\includegraphics[width=\textwidth]{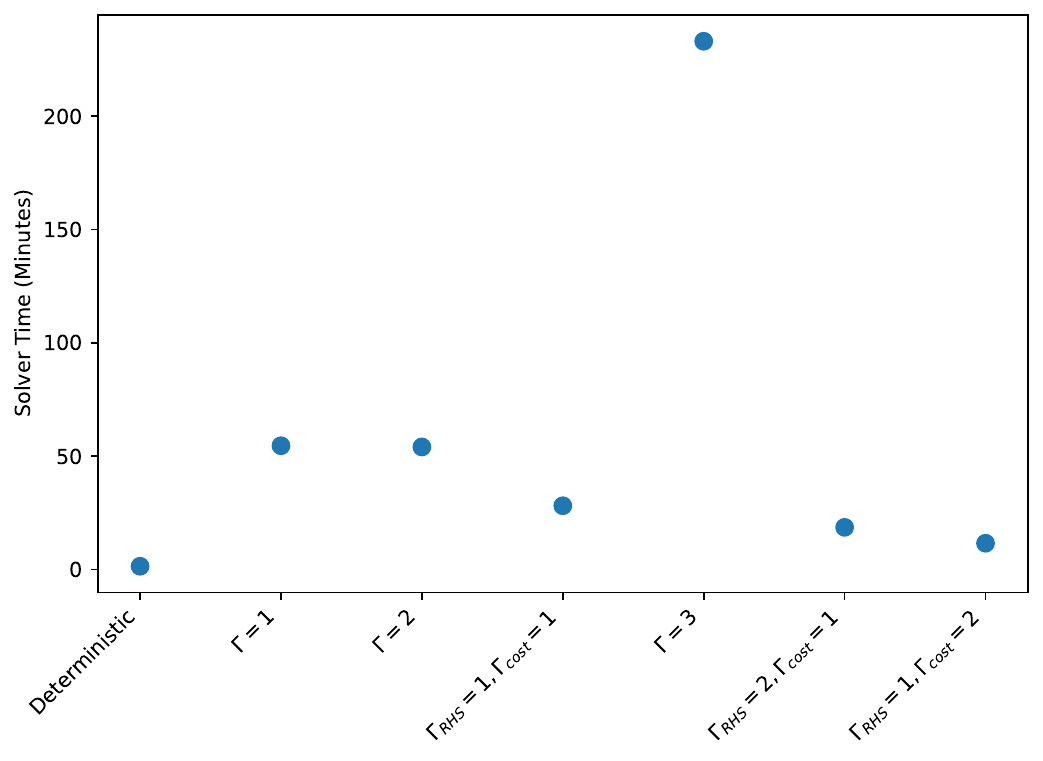}} 
{Solution time by scenario. \label{fig:solver_time}}
{}
\end{figure}

Figures \ref{fig:formulation_comparison_maps} and \ref{fig:formulation_comparison_upsides} in the Appendix show a comparison of the split- and combined-budget formulations in terms of geospatial transmission build, and upside performance. Additionally, Figures \ref{fig:sanity_main}, \ref{fig:sanity_maps}, \ref{fig:sanity_second_stage}, and \ref{fig:sanity_upside} in the Appendix compare the split-budget results to cases where the split budget is reproduced with scenarios alone. These cases demonstrate that the split-budget results are a result of the uncertainty budget being split and not a result of an unintended side effect of the dualized split-budget formulation.

\section{Concluding remarks}\label{sec:Conclusion}

In this section we highlight several key findings of our analysis:

\begin{enumerate}
	\item In the context of the capacity expansion problem for the State of California, using a robust formulation as opposed to the standard deterministic formulation leads to the model selecting significantly more in-state transmission, highlighting transmission as a key lever for increasing robustness and adaptability to uncertainty over the long-term planning horizon. We expect that other jurisdictions may similarly find markedly different portfolios as a result of using robust formulations for planning; in general, the value of robust planning is highest when the decision space is large and when there are certain investments available that can help to hedge against uncertainty.
	\item This increased investment in transmission increases near-term costs slightly, but offers significant long-term benefits in the form of potential avoided ``downside'' scenarios where key uncertainties turn out much worse than expected. In this way, the added robustness can be seen as an ``insurance policy'' that helps avoid severe downside risks. Notably, our analysis finds that a significant amount of robustness can be added for a negligible increase in near-term costs.
	\item The portfolios selected by the robust model, while helping to avoid major downside cases, do not sacrifice on ``upside'' performance in cases where uncertain parameters turn out better than expected. While we cannot ensure this finding generalizes to all circumstances, it does show that avoiding worst-case outcomes through robust optimization does not necessarily come at the cost of missing out on potential cost savings under more positive outcomes.
    \item The current deterministic planning processes used for generation and transmission planning are likely to be exposing the system and ratepayers to significant outage and affordability risks, particularly in California as demonstrated in this paper, which could be mitigated via the use of robust and/or stochastic formulations like the one discussed in this paper.
	\item The split-budget formulation introduced in this paper has the impact of modeling a subset of the uncertainty set that is modeled under the standard ARO formulation. Specifically, the formulation models scenarios at optimality where both one or more cost parameters and one or more right-hand side (RHS) constraint parameters take on their worst-case outcomes. This reduces robustness slightly, compared to robust optimization considering the full uncertainty set covering all possible combinations of cost and RHS parameter realizations. The split-budget formulation has significant advantages in terms of computational tractability, allowing the endogenization of uncertainty into large-scale, real-world planning problems, for which traditional ARO or stochastic optimization would be unlikely to be tractable.
\end{enumerate}

The split-budget method introduced in this paper could be introduced into real-world planning processes without modification. Such a model formulation does not, however, have to replace existing deterministic models completely. Rather, in the near term we recommend a parallel approach, in which robust models could be used alongside their deterministic counterparts, highlighting how the portfolios they produce are different. Ideally, both types of models would be compared within a strong scenario analysis framework (as described in \cite{mantegna_integrated_2026}), in which a range of portfolios is developed, and then stress-tested over a range of key uncertain outcomes, similar to the stress-testing section of this paper. This would highlight the potential future benefits of the more robust portfolios, and allow decision-makers to make an informed decision regarding how much ``robustness'' is appropriate to add. In this manner, decision-makers can gain confidence with robust formulations, and experience their benefits, without abandoning existing models and processes completely.


%

\clearpage

\begin{APPENDIX}{Additional Results}
\FloatBarrier

\begin{figure}
     \FIGURE
    {\includegraphics[width=\textwidth]{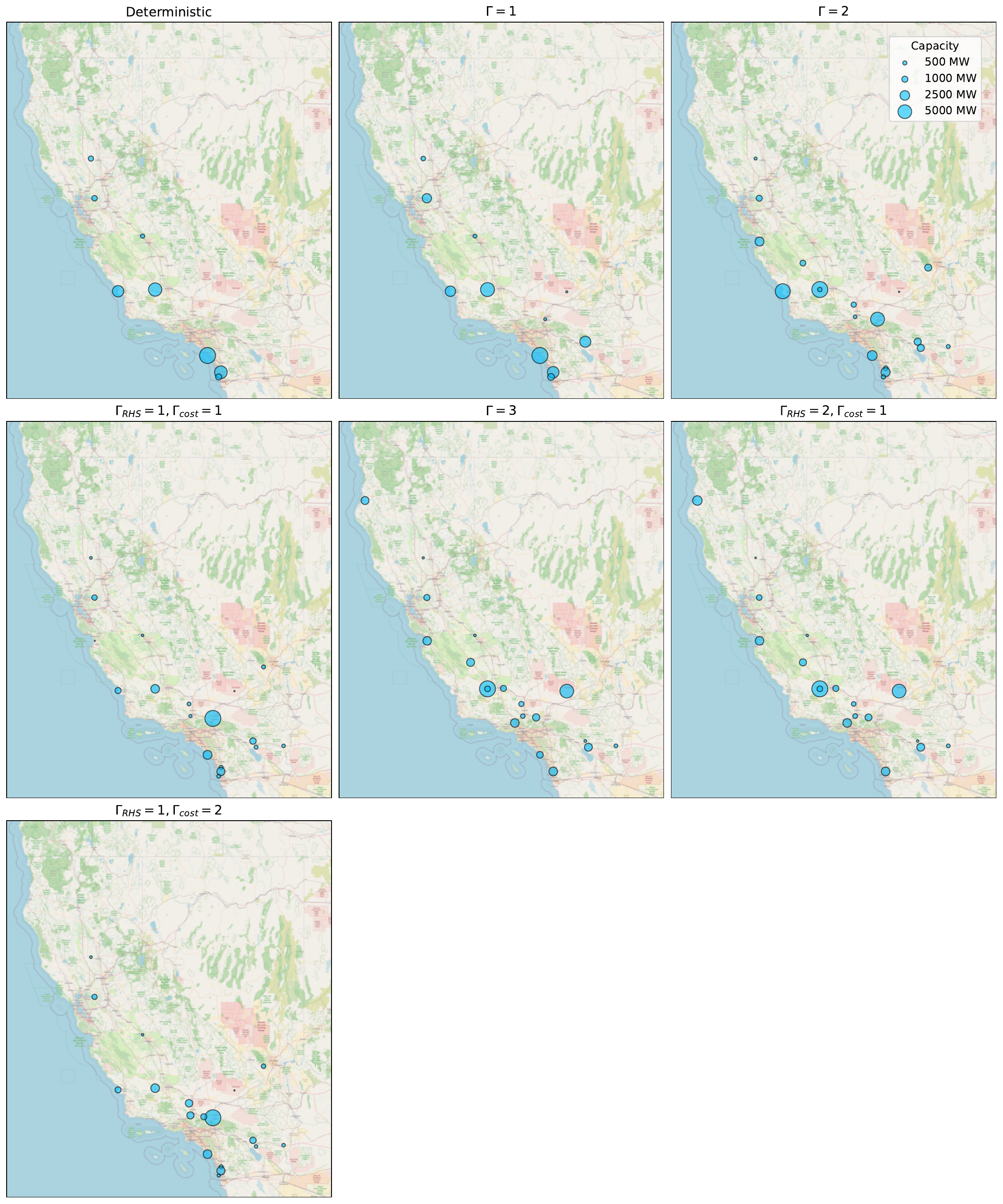}} 
{Comparison of geographical transmission upgrades selected between standard and split-budget robust formulations. \label{fig:formulation_comparison_maps}}
{}
\end{figure}

\begin{figure}
     \FIGURE
    {\includegraphics[width=\textwidth]{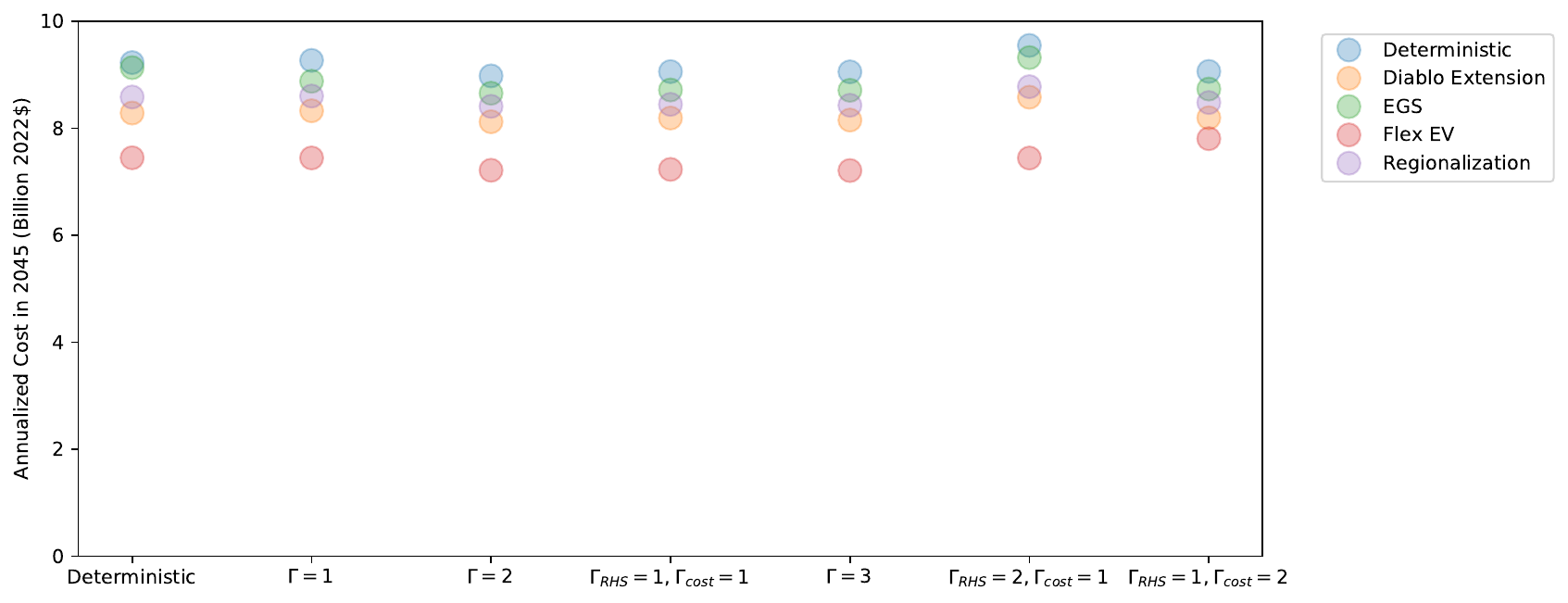}} 
{Comparison of upside performance between standard and split-budget robust formulations. \label{fig:formulation_comparison_upsides}}
{}
\end{figure}

\begin{figure}
     \FIGURE
    {\includegraphics[width=\textwidth]{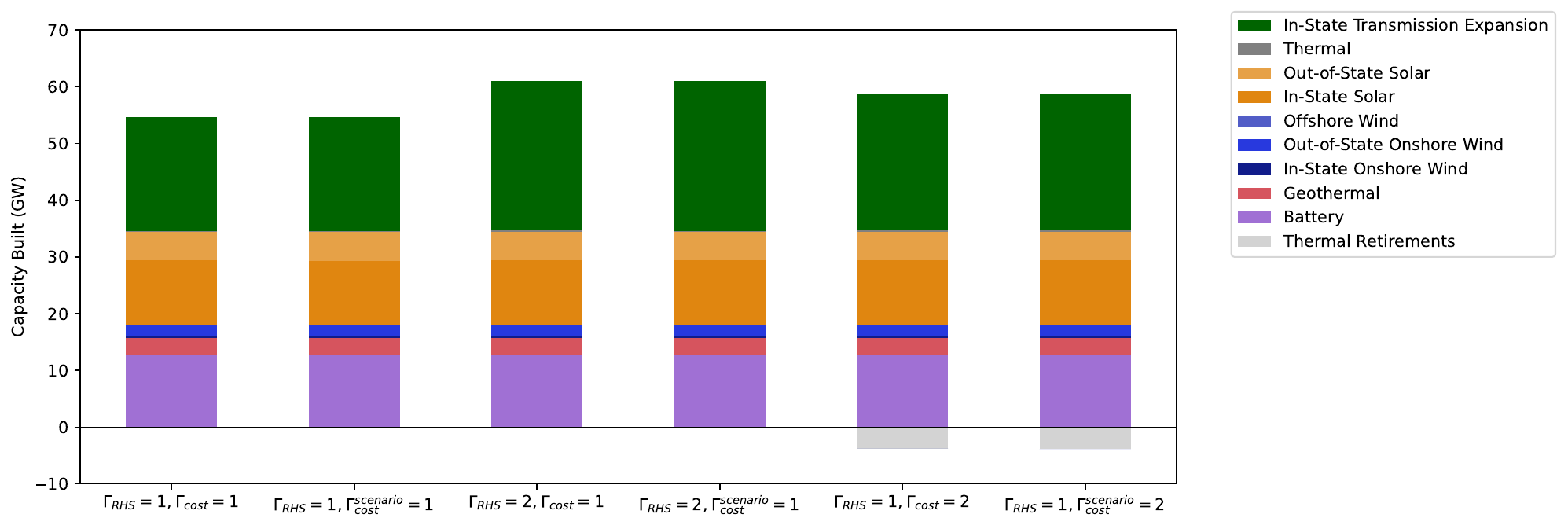}} 
{Comparison of first-stage resource builds between split-budget formulation and its corresponding scenario-based counterpart. \label{fig:sanity_main}}
{}
\end{figure}

\begin{figure}
     \FIGURE
    {\includegraphics[width=0.75\textwidth]{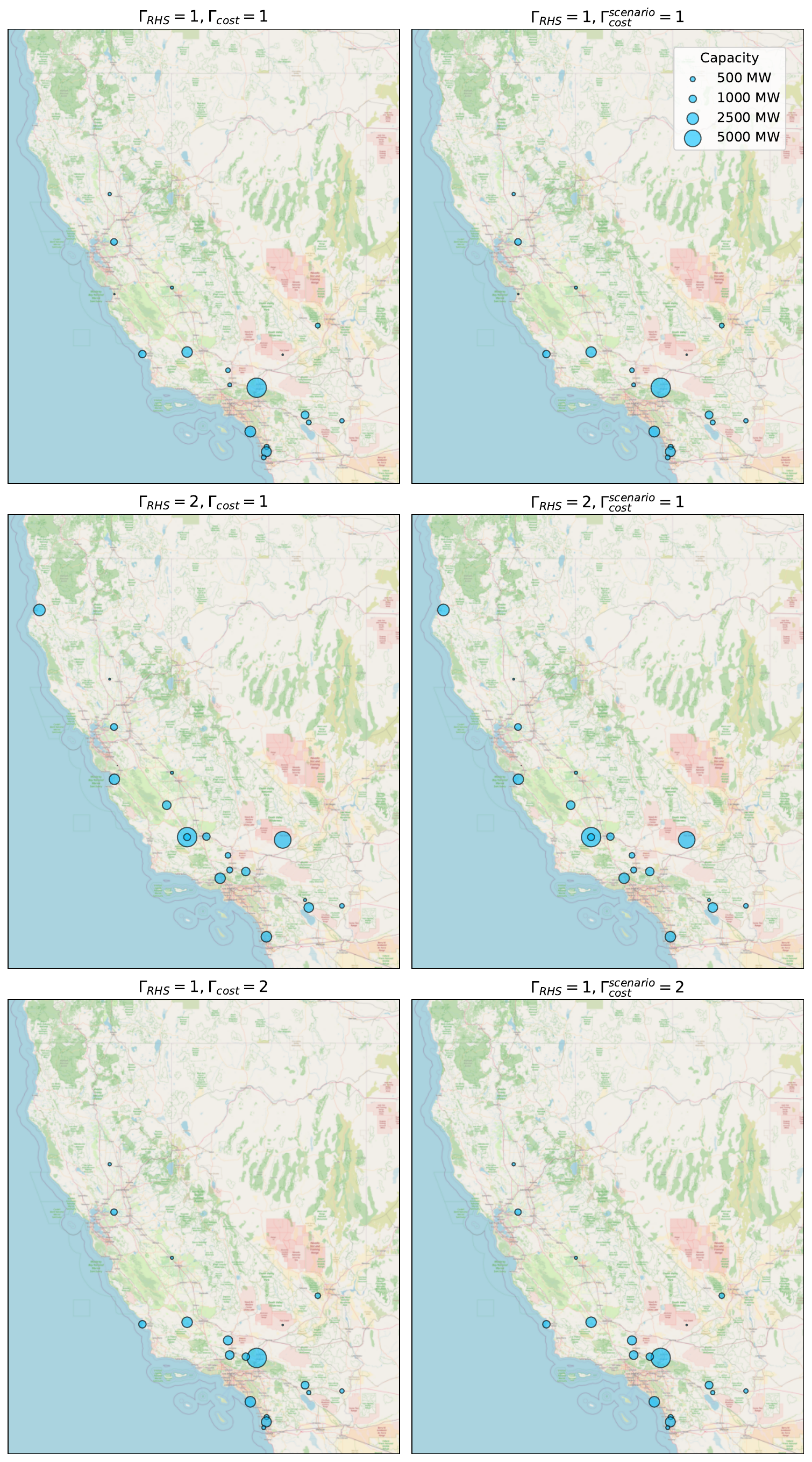}} 
{Comparison of first-stage transmission decisions between split-budget formulation and its corresponding scenario-based counterpart. \label{fig:sanity_maps}}
{}
\end{figure}

\begin{figure}
     \FIGURE
    {\includegraphics[width=\textwidth]{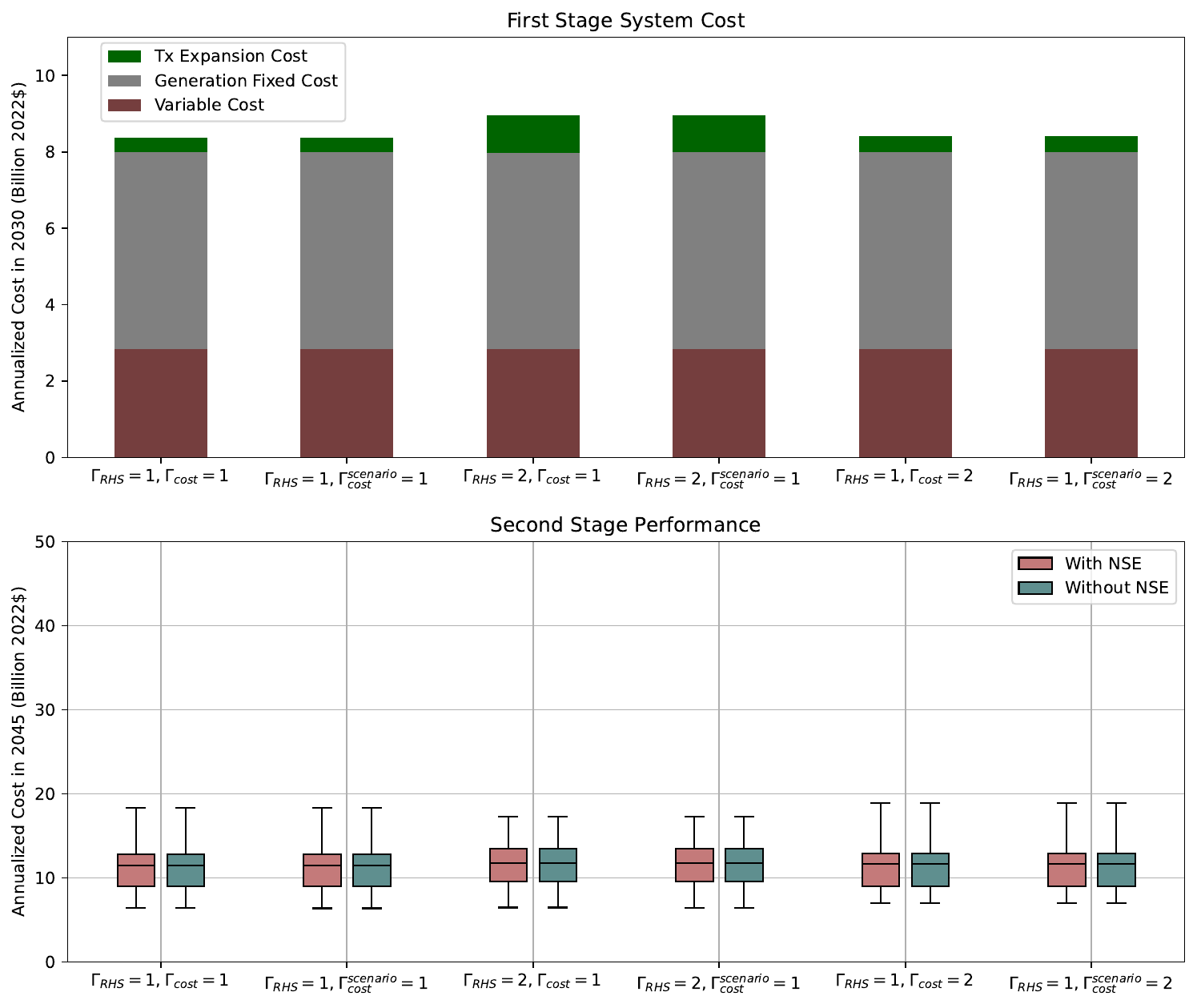}} 
{Comparison of portfolio costs between split-budget formulation and its corresponding scenario-based counterpart. \label{fig:sanity_second_stage}}
{}
\end{figure}

\begin{figure}
     \FIGURE
    {\includegraphics[width=\textwidth]{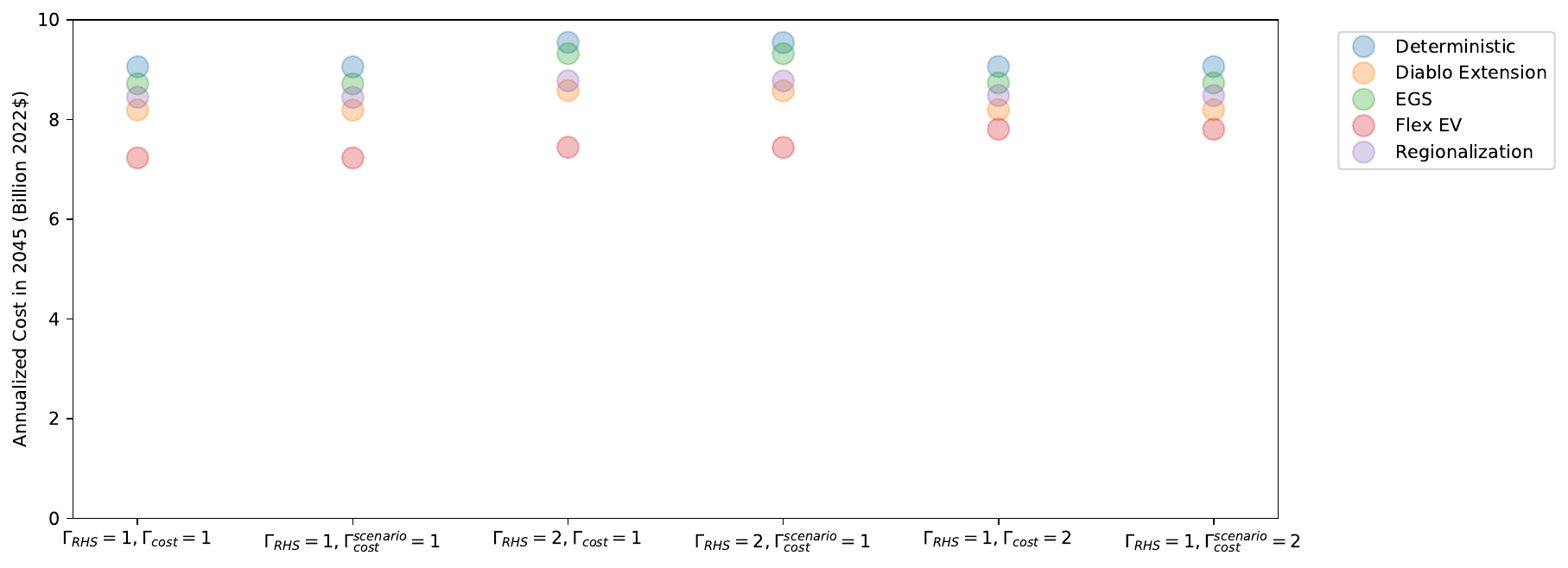}} 
{Comparison of upside performance between split-budget formulation and its corresponding scenario-based counterpart. \label{fig:sanity_upside}}
{}
\end{figure}

\end{APPENDIX}
%
%


\clearpage

\ACKNOWLEDGMENT{We would like to express our sincere gratitude to Ryan Tracey, Amit Ranjan, Jim Himelic, Jaxon Stuhr, Sara Maatta, and Spandan Gandhi for their support in conceptualizing and developing this research project. We would also like to thank Prof. Ben Hobbs for his foundational work on planning under uncertainty that this work builds on, Prof. Warren Powell for his work on sequential decision making under uncertainty that helped the authors formulate the research question, and Prof. Dimitris Bertsimas for his work on adaptive robust optimization that this work directly builds on.}


\bibliographystyle{informs2014} 
\bibliography{references.bib} 





\end{document}